\newif\ifjourn
\ifjourn
    \documentclass[pdflatex,sn-mathphys-num]{sn-jnl}
    \theoremstyle{thmstyleone}
\else
    \documentclass{article}

    \usepackage[english]{babel}
    \usepackage[letterpaper,top=2cm,bottom=2cm,left=3cm,right=3cm,marginparwidth=1.75cm]{geometry}

    \usepackage[numbers]{natbib}
    \usepackage[colorlinks=true, allcolors=blue]{hyperref}
    \author{Matthias Bentert \and Fedor V. Fomin \and Petr A. Golovach}
\fi

\usepackage{mathdots} 
\usepackage{amsthm}
\usepackage{amsmath}
\usepackage{dsfont}
\usepackage{graphicx}
\usepackage{nicefrac}
\usepackage{tikz}
\usetikzlibrary{calc,math,patterns,shapes,backgrounds}
\usepackage{todonotes} 
\setuptodonotes{inline}
\usepackage{tabularx}
\usepackage[boxed,noline,noend,ruled,linesnumbered]{algorithm2e}
\usepackage{xintexpr}
\usepackage{boxedminipage}
\usepackage{framed}
\usepackage{tabularx}
\usepackage{tcolorbox}
\usepackage{caption}
\usepackage{xifthen}
\usepackage{thm-restate}
\usepackage[capitalize,noabbrev]{cleveref}
        
\newtheorem{theorem}{Theorem}

\newtheorem{corollary}{Corollary}
\newtheorem{open problem}{Open Problem}

\newcommand{\disp}{\textsc{Vertex-Disjoint Paths}}
\newcommand{\mdp}{\textsc{Maximum Vertex-Disjoint Paths}}
\newcommand{\dsp}{\textsc{Vertex-Disjoint Shortest Paths}}
\newcommand{\mdsp}{\textsc{Maximum Vertex-Disjoint Shortest Paths}}
\newcommand{\ldsp}{\textsc{Layered Vertex-Disjoint Shortest Paths}}

\newcommand{\npconp}{NP $\subseteq$ coNP/poly}

\DeclareMathOperator{\poly}{poly}
\DeclareMathOperator{\dist}{dist}
\DeclareMathOperator{\opt}{OPT}

\tikzset{vertex/.style ={circle,inner sep=0pt,minimum size=4pt}}
\tikzstyle{path} = [color=black!15,line cap=round, line join=round, line width=6pt]

\newlength{\RoundedBoxWidth}
\newsavebox{\GrayRoundedBox}
\newenvironment{GrayBox}[1]%
   {\setlength{\RoundedBoxWidth}{.93\textwidth}
    \def\boxheading{#1}
    \begin{lrbox}{\GrayRoundedBox}
       \begin{minipage}{\RoundedBoxWidth}}%
   {   \end{minipage}
    \end{lrbox}
    \begin{center}
    \begin{tikzpicture}%
       \node(Text)[draw=black!20,fill=white,rounded corners,%
             inner sep=2ex,text width=\RoundedBoxWidth]%
             {\usebox{\GrayRoundedBox}};
        \coordinate(x) at (current bounding box.north west);
        \node [draw=white,rectangle,inner sep=3pt,anchor=north west,fill=white] 
        at ($(x)+(6pt,.75em)$) {\boxheading};
    \end{tikzpicture}
    \end{center}}
    
\newenvironment{defproblemx}[2][]{\noindent\ignorespaces%
                                \FrameSep=6pt%
                                \parindent=0pt%
                \vspace*{-1.5em}
                \ifthenelse{\isempty{#1}}{%
                  \begin{GrayBox}{\textsc{#2}}%
                }{%
                  \begin{GrayBox}{\textsc{#2}  parameterized by~{#1}}%
                }
                \begin{tabular*}{\textwidth}{@{\hspace{.1em}} >{\itshape} p{1.8cm} p{0.8\textwidth} @{}}%
            }{
                \end{tabular*}%
                \end{GrayBox}%
                \ignorespacesafterend
            }  

\newcommand{\defproblem}[3]{
  \begin{defproblemx}{#1}
    Input:  & #2 \\
    Question: & #3
  \end{defproblemx}
}%

\title{Tight Approximation and Kernelization Bounds for Vertex-Disjoint Shortest Paths}
\date{}

\ifjourn
\fi

\begin{document}
\ifjourn
\title[Tight Approximation and Kernelization Bounds for Vertex-Disjoint Shortest Paths]{Tight Approximation and Kernelization Bounds for Vertex-Disjoint Shortest Paths\footnote{A preliminary version of this paper was presented at the 42nd International Symposium on Theoretical Aspects of Computer Science (STACS~'25)~\cite{BFG25}. This version contains additional hardness results regarding subexponential-time algorithms and planar graphs (Propositions~1~and~2).}}

\author{\fnm{Matthias} \sur{Bentert}}\email{matthias.bentert@uib.no}

\author{\fnm{Fedor V.} \sur{Fomin}}\email{fedor.fomin@uib.no}

\author*{\fnm{Petr A.} \sur{Golovach}}
\email{petr.golovach@uib.no}

\affil{\orgname{University of Bergen}, \orgaddress{ \city{Bergen}, \country{Norway}}}

\abstract{
We examine the possibility of approximating \mdsp. In this problem,
the input is an edge-weighted (directed or undirected) $n$-vertex graph~$G$ along with~$k$~terminal pairs~$(s_1,t_1),(s_2,t_2),\ldots,(s_k,t_k)$. The task is to connect as many terminal pairs as possible by pairwise vertex-disjoint paths such that each path is a shortest path between the respective terminals. 
Our work is anchored in the recent breakthrough by Lochet [SODA '21], which demonstrates the polynomial-time solvability of the problem for a fixed value of $k$. 

Lochet's result implies the existence of a polynomial-time $ck$-approximation for \mdsp, where $c \leq 1$ is a constant. (One can guess~$\nicefrac{1}{c}$ terminal pairs to connect in~$k^{O(\nicefrac{1}{c})}$~time and then utilize Lochet's algorithm to compute the solution in~$n^{f(\nicefrac{1}{c})}$~time.)
Our first result suggests that this approximation algorithm is, in a sense, the best we can hope for.
More precisely, assuming the gap-ETH, we exclude the existence of an~\mbox{$o(k)$-approximation} within~$f(k)\poly(n)$~time for any function~$f$ that only depends on~$k$.

Our second result demonstrates the infeasibility of achieving an approximation ratio of $m^{\nicefrac{1}{2}-\varepsilon}$ in polynomial time, unless P $=$ NP. It is not difficult to show that a greedy algorithm selecting a path with the minimum number of arcs results in a $\lceil\sqrt{\ell}\rceil$-approximation, where~$\ell$~is the number of edges in all paths of an optimal solution. Since $\ell \leq n$, this underscores the tightness of the~$m^{\nicefrac{1}{2}-\varepsilon}$-inapproximability bound. 
Additionally, we establish that \mdsp{} can be solved in~$2^{O(\ell)} \poly(n)$ time, but does not admit a polynomial kernel in~$\ell$.
Moreover, it cannot be solved in~$2^{o(\ell)}\poly(n)$ time.
Our hardness results hold for undirected graphs with unit weights, while our positive results extend to scenarios where the input graph is directed and features arbitrary (non-negative) edge weights.
}


\fi
\maketitle

\ifjourn
\else
\begin{abstract}
We examine the possibility of approximating \mdsp. In this problem,
the input is an edge-weighted (directed or undirected) $n$-vertex graph~$G$ along with~$k$~terminal pairs~$(s_1,t_1),(s_2,t_2),\ldots,(s_k,t_k)$. The task is to connect as many terminal pairs as possible by pairwise vertex-disjoint paths such that each path is a shortest path between the respective terminals. 
Our work is anchored in the recent breakthrough by Lochet [SODA '21], which demonstrates the polynomial-time solvability of the problem for a fixed value of $k$. 

Lochet's result implies the existence of a polynomial-time $ck$-approximation for \mdsp, where $c \leq 1$ is a constant. (One can guess~$\nicefrac{1}{c}$ terminal pairs to connect in~$k^{O(\nicefrac{1}{c})}$~time and then utilize Lochet's algorithm to compute the solution in~$n^{f(\nicefrac{1}{c})}$~time.)
Our first result suggests that this approximation algorithm is, in a sense, the best we can hope for.
More precisely, assuming the gap-ETH, we exclude the existence of an~\mbox{$o(k)$-approximation} within~$f(k)\poly(n)$~time for any computable function~$f$ that only depends on~$k$.

Our second result demonstrates the infeasibility of achieving an approximation ratio of $m^{\nicefrac{1}{2}-\varepsilon}$ in polynomial time, unless P $=$ NP. It is not difficult to show that a greedy algorithm selecting a path with the minimum number of arcs results in a $\lceil\sqrt{\ell}\rceil$-approximation, where~$\ell$~is the number of edges in all paths of an optimal solution. Since $\ell \leq \min(n,m)$, this underscores the tightness of the~$m^{\nicefrac{1}{2}-\varepsilon}$-inapproximability bound. 
Additionally, we establish that \mdsp{} can be solved in~$2^{O(\ell)} \poly(n)$ time, but does not admit a polynomial kernel in~$\ell$.
Moreover, it cannot be solved in~$2^{o(\ell)}\poly(n)$ time unless the exponential time hypothesis fails.
Our hardness results hold for undirected graphs with unit weights, while our positive results extend to scenarios where the input graph is directed and features arbitrary (non-negative) edge weights.
\end{abstract}
\fi

\section{Introduction}
We study a variant of the well-known problem \disp. In the latter, the input comprises a  (directed or undirected) graph~$G$ and~$k$ terminal pairs. The task is to identify whether pairwise vertex-disjoint paths can connect all terminals. \disp{} has long been established as NP-complete \cite{Kar75} and has played a pivotal role in the graph-minor project by Robertson and Seymour \cite{RS95}.

Eilam-Tzoreff \cite{Eilam-Tzoreff98} introduced a variant of \disp{} where all paths in the solution must be \emph{shortest} paths between the respective terminals. The parameterized complexity of this variant,  known as \dsp,  was recently resolved~\cite{Loc21} and subsequently the running time improved~\cite{BNRZ23}: The problem, parameterized by $k$, is W[1]-hard and in~XP (that is, polynomial-time solvable for any constant~$k$) for undirected graphs. 
On directed graphs, the problem is NP-hard already for $k=2$ if zero-weight edges are allowed \cite{FHW80}. The problem is solvable in polynomial time for~$k=2$ for positive edge weights \cite{BK17}. It is NP-hard when~$k$ is part of the input and the complexity for constant~$k > 2$ remains open.

The optimization variant of \dsp, where not necessarily all terminal pairs need to be connected, but at least $p$ of them, is referred to as \mdsp.

\defproblem{\mdsp}
{A graph~$G=(V,E)$, an edge-length function~$w \colon E \rightarrow \mathds{Q}_{\geq 0}$, terminal pairs~$(s_1,t_1), (s_2,t_2), \ldots, (s_k,t_k)$ where $s_i\neq t_i$ for each~$i\in[k]$, and an integer~$p$.}
{Is there a set~$S \subseteq [k]$ with~$|S| \geq p$ such that there is a collection~${\mathcal{C}=\{P_i\}_{i\in S}}$ of pairwise vertex-disjoint paths such that 
for each~$i \in S$,  path $P_i$ is a shortest path from~$s_i$ to~$t_i$?}


A few remarks are in order. In the literature concerning \disp{} and its variants, one usually distinguishes between vertex-disjoint and internally vertex-disjoint paths.
In the latter, two paths in a solution might share common endpoints while in the former, all paths must be completely vertex disjoint---including the two ends.
We focus on the variant where paths must be completely vertex-disjoint, but most of our results also hold for internally vertex-disjoint paths.

Note that \dsp{} is a special case of \mdsp{} with~$p=k$. 
For the maximization version, we are not given $p$ as input but are instead asked to find a set $S$ that is as large as possible. Slightly abusing notation, we do not distinguish between these two variants and refer to both as \mdsp.

While the parameterization by $k$ yields strong hardness results (both in terms of parameterized complexity and, as we will show later, approximation), another natural parameterization would be the sum of path lengths in a solution.
We initiate the study of a related parameter~$\ell$, the minimum number of edges in an optimal solution (assuming the instance is a yes-instance, otherwise, we define~$\ell=n$).
If we confine all edge weights to be positive integers, then $\ell$ serves as a lower bound for the sum of path lengths. Since our hardness results apply to unweighted graphs, studying~$\ell$ instead of the sum of path lengths does not weaken the negative results and~$\ell$ proves to be very useful for approximation and parameterized algorithms.
Note that the sum of path lengths is not a suitable parameter as dividing all edge weights by~$m \cdot w_{\max}$ (where~$w_{\max}$ is the maximum weight of any edge in the input) yields an equivalent instance where the sum of path lengths in the solution is at most one.

For the parameterized complexity of \mdsp, we note that the results for \dsp{}~\cite{BNRZ23,Loc21} for the parameterization by $k$ directly translate for \mdsp{} parameterized by $p$. The problem is~W[1]-hard as a generalization of \dsp{}, and to obtain an XP algorithm, it is sufficient to observe that in $n^{O(p)}$ time we can guess a set $S\subseteq[k]$ of size $p$ and apply the XP algorithm for \dsp{} for the selected set of terminal pairs.     

Before the recent work of Chitnis, Thomas, and Wirth~\cite{CTW24}, little was known about approximation algorithms for \mdsp. Chitnis, Thomas, and   Wirth demonstrated that no~$(2-\varepsilon)$-approximation could be achieved in time $f(k) n^{o(k)}$  assuming the gap-ETH. 

For the related \mdp, where the task is to connect the maximum number of terminal pairs by disjoint but not necessarily shortest paths, 
 $O(\sqrt{n})$-approximation algorithms are known~\cite{Klein96,KS98}. The best known lower bounds for this variant are~$2^{\Omega(\sqrt{\log n})}$ and~$n^{\Omega(\nicefrac{1}{(\log \log n)^2})}$. 
The first lower bound 
 holds even if the input graph is an unweighted planar graph, while the second holds even if the input graph is an unweighted grid graph \cite{CKN21,CKN22}. For these two special cases, there are approximation algorithms achieving ratios $\tilde{O}(n^{\nicefrac{9}{19}})$ and~$\tilde{O}(n^{\nicefrac{1}{4}})$, respectively~\cite{CKN21,CKN22}.

When requiring the solution paths to be edge-disjoint rather than vertex-disjoint, it is known that even computing a~$m^{\nicefrac{1}{2}-\varepsilon}$-approximation is NP-hard in the directed setting~\cite{GKRSY03}.
There have also been some studies on relaxing the notion so that each edge appears in at most~$c > 1$ of the solution paths.  
The integer~$c$ is called the congestion and the currently best known approximation algorithm achieves a $\poly(\log n)$-approximation with~$c=2$~\cite{CL16}.

\medskip\noindent\textbf{Our results.}
We show that computing a~$m^{\nicefrac{1}{2}-\varepsilon}$-approximation for \mdsp{} is NP-hard for any~$\varepsilon >0$ (\Cref{thm:nosqrt}).
Moreover in terms of FPT-approximations, we demonstrate in \Cref{thm:noapprox} that any~$k^{o(1)}$-approximation in time ${f(k)\poly(n)}$ implies FPT~$=$~W[1]
and that it is impossible to achieve an $o(k)$-approximation in time ${f(k)\poly(n)}$  unless the gap-ETH fails.
This significantly improves the current state of approximation results for  \mdsp{}  in two ways.
First, we use the weaker assumption~FPT~$\neq$~W[1] instead of the gap-ETH. Second, our theorem excludes approximation factors polynomial in the input size, rather than only constant factors larger than 2 as shown by  Chitnis et al.~\cite{CTW24}.

We complement the first lower bound by showing that a simple greedy strategy for \textsc{Maximum Vertex-Disjoint Paths} achieves a~$\lceil\sqrt{\ell}\rceil$-approximation for \mdsp{} (\Cref{prop:lapprox}).
In \Cref{prop:fptl,thm:nopkl}, we show that \mdsp{} is fixed-parameter tractable when parameterized by~$\ell$, but it does not admit a polynomial kernel.
The algorithm runs in~$2^{O(\ell)}\poly(n)$ time and we also exclude~$2^{o(\ell)}\poly(n)$ algorithms under the ETH in \cref{thm:main-ETH}.
Interestingly, our reduction also excludes~$2^{o(\sqrt{n})}$-time algorithms for planar input graphs.
We mention that our hardness results hold for undirected graphs with unit weights, and all our positive results hold even for directed and edge-weighted input graphs.
We summarize our results  in~\cref{tab:results}.
\begin{table}[t]
\centering
\caption{Overview of our results. New results are bold. All hardness results hold for unweighted and undirected graphs, while all new algorithmic results hold even for directed graphs with arbitrary non-negative edge weights.}
\begin{tabular}{c|c|c}
Parameter & Exact & Approximation \\\hline
no & NP-complete and \textbf{no $2^{o(n+m)}$} & \textbf{no $\mathbf{m^{\nicefrac{1}{2}-\varepsilon}}$-approximation in $\text{poly}(n)$ time} \\\hline
$k$ & XP and W[1]-hard & \textbf{no $\mathbf{o(k)}$-approximation in~$\mathbf{f(k)\text{poly}(n)}$ time} \\\hline
$\ell$ & \textbf{FPT} and \textbf{no poly kernel}  & \textbf{$\mathbf{\lceil\sqrt{\ell}\rceil}$-approximation} 
\end{tabular}
\label{tab:results}
\end{table}


\section{Preliminaries}
\label{sec:prelim}
For a positive integer~$x$, we denote by~$[x] = \{1,2,\dots,x\}$ the set of all positive integers at most~$x$.
We denote by~$G = (V,E)$ a graph and by~$n$ and~$m$ the number of vertices and edges in~$G$, respectively.
A graph~$G$ is said to be~\emph{$k$-partite} if~$V$ can be partitioned into~$k$ disjoint sets~$V_1,V_2,\ldots,V_k$ such that each set~$V_i$ induces an independent set, that is, there is no edge~$\{u,v\}\in E$ with~$\{u,v\}\subseteq V_i$ for some~$i \in [k]$.
The \emph{degree} of a vertex~$v$ is the number of edges in~$E$ that contain~$v$ as an endpoint and the \emph{maximum degree} of a graph is the highest degree of any vertex in the graph.

A \emph{path} in a graph~$G$ is a sequence~$(v_0,v_1,\dots,v_\ell)$ of distinct vertices such that each pair~$(v_{i-1},v_i)$ is connected by an edge in~$G$.
The first and last vertex $v_0$ and $v_\ell$ are called the \emph{ends} of $P$.
We also say that~$P$ is a path \emph{from} $v_0$ \emph{to} $v_\ell$ or a $v_0$-$v_\ell$-path.
The length of a path is the sum of its edge lengths or simply the number~$\ell$ of edges if the graph is unweighted.
For two vertices~$v,w$, we denote the length of a shortest $v$-$w$-path in~$G$ by~$\dist_G(v,w)$ or~$\dist(v,w)$ if the graph~$G$ is clear from the context.


We assume the reader to be familiar with the big-O notation and basic concepts in computational complexity like NP-completeness and reductions.
We refer to the textbook by Garey and Johnson \cite{GJ79} for an introduction.
Throughout this paper, we reduce from~\textsc{3-Sat}, \textsc{Clique}, and~\textsc{Multicolored Clique}, three of the most fundamental problems in theoretical computer science.
We state their definitions for the sake of completeness.

\defproblem{3-Sat}
{A Boolean formula~$\phi$ in conjunctive normal form where each clause contains at most three literals.}
{Is $\phi$ satisfiable?}

\defproblem{Clique}
{An undirected graph~$G$ and an integer~$k$.}
{Is there a clique of size~$k$ in~$G$, that is, a set of~$k$ vertices that are pairwise neighbors?}

\defproblem{Multicolored Clique}
{An undirected $x$-partite graph~$G$ and an integer~$k \leq x$.}
{Is there a clique of size~$k$ in~$G$?}

It is more common to state \textsc{Multicolored Clique} for $x=k$ and, in this case, the partitions of the input graph are often modelled as colors and a clique is called multicolored as it contains exactly one vertex from each color class.
However, it is convenient for us to allow~$k\leq x$ and we call a clique in this context \emph{multicolored} if it contains at most one vertex from each color class.

For a detailed introduction to parameterized complexity and kernelization, we refer the reader to the textbooks by Cygan et al.~\cite{CyganFKLMPPS15} and Fomin et al.~\cite{FominLSZ19}.
A \emph{parameterized problem}~$P$ is a language containing pairs~$(I,\rho)$ where~$I$ is an instance of an (unparameterized) problem and~$\rho$ is an integer called the \emph{parameter}.
In this paper, the parameter will usually be either the number~$k$ of terminal pairs or the minimum number~$\ell$ of edges in a solution (a maximum collection of vertex-disjoint shortest paths between terminal pairs).
A parameterized problem~$P$ is \emph{fixed-parameter tractable} if there exists an algorithm solving any instance~$(I,\rho)$ of~$P$ in~$f(\rho)\poly(|I|)$ time, where~$f$ is some computable function only depending on~$\rho$.
To show that a problem is presumably not fixed-parameter tractable, one usually shows that the problem is hard for a complexity class known as W[1].
The class \emph{XP} contains all parameterized problems which can be solved in~$|I|^{f(\rho)}$ time, that is, in polynomial time if~$\rho$ is constant.
A parameterized problem is said to admit a \emph{polynomial kernel}, if there is a polynomial-time algorithm that given an instance~$(I,\rho)$ computes an equivalent instance~$(I',\rho')$ (called the \emph{kernel}) such that~$|I'|+\rho'$ is upper-bounded by a polynomial in~$\rho$.
It is known that any parameterized problem admitting a polynomial kernel is fixed-parameter tractable and each fixed-parameter tractable problem is contained in XP.

An~$\alpha$-approximation algorithm for a maximization problem is a polynomial-time algorithm that for any input returns a solution of size at least~$\nicefrac{\opt}{\alpha}$ where~$\opt$ is the size of an optimal solution.
A parameterized~$\alpha$-approximation algorithm also returns a solution of size at least~$\nicefrac{\opt}{\alpha}$, but its running time is allowed to be~$f(\rho)\poly(n)$, where~$\rho$ is the parameter and~$f$ is some computable function only depending on~$\rho$.
In this work, we always consider (unparameterized) approximation algorithms unless we specifically state a parameterized running time.

To exclude an~$\alpha$-approximation for an optimization problem, one can use the framework of \emph{approxi\-mation-preserving reductions}.
A strict approximation-preserving reduction is a pair of algorithms---called the \emph{reduction algorithm} and the \emph{solution-lifting algorithm}---that both run in polynomial time and satisfy the following.
The reduction algorithm takes as input an instance~$I$ of a problem~$L$ and produces an instance~$I'$ of a problem~$L'$.
The solution-lifting algorithm takes any solution~$S'$ of~$I'$ and transforms it into a solution~$S$ of~$I$ such that if~$S'$ is an~$\alpha$-approximation for~$I'$ for some~$\alpha \geq 1$, then~$S$ is an~$\alpha$-approximation for~$I$.
If a strict approximation-preserving reduction from~$L$ to~$L'$ exists and~$L$ is hard to approximate within some value~$\beta$, then~$L'$ is also hard to approximate within~$\beta$.

The \emph{exponential time hypothesis (ETH)} introduced by Impagliazzo and Paturi \cite{IP01} states that there is some~$\varepsilon > 0$ such that each (unparameterized) algorithm solving \textsc{3-Sat} takes at least~$2^{\varepsilon n + o(n)}$~time, where~$n$ is the number of variables in the input instance.
A stronger conjecture called the \emph{gap-ETH} was independently introduced by Dinur~\cite{Din16} and Manurangsi and Raghavendra~\cite{MR17}.
It states that there exist~$\varepsilon, \delta > 0$ such that any $(1-\varepsilon)$-approximation algorithm for \textsc{Max 3-Sat}\footnote{\textsc{Max 3-Sat} is a generalization of \textsc{3-Sat} where the question is not whether the input formula is satisfiable but rather how many clauses can be satisfied simultaneously.} takes at least~$2^{\delta n+o(n)}$ time.

\section{Approximation}
\label{sec:approx}
In this section, we show that \mdsp{} cannot be~$o(k)$-approx\-imated in~$f{(k)\poly(n)}$~time unless the gap-ETH fails and no~$m^{\nicefrac{1}{2}-\varepsilon}$-approximation in polynomial time unless P~$=$~NP.
We complement the latter result by developing a~$\lceil\sqrt{\ell}\rceil$-approximation algorithm that runs in polynomial time.
We start with a reduction based on a previous reduction by Bentert et al.~\cite{BNRZ23} and make it approximation-preserving.\footnote{We mention in passing that essentially the same modification to the reduction by Bentert et al. has been found by Akmal et al.~\cite{AVW24} in independent research. While we use the modification to show stronger inapproximability bounds, they use it to show stronger fine-grained hardness results with respect to the minimum degree of a polynomial-time algorithm for constant~$k$.}
Moreover, our result is tight in the sense that a~$k$-approximation can be computed in polynomial time by simply connecting any terminal pair by a shortest path.
A~$ck$-approximation for any constant~$c \leq 1$ can also be computed in polynomial time by guessing~$\frac{1}{c}$ terminal pairs to connect and then using the XP-time algorithm by Bentert et al.~\cite{BNRZ23} to find a solution.
Note that since~$c$ is a constant, the XP-time algorithm for~$\frac{1}{c}$~terminal pairs runs in polynomial time.

\begin{theorem}
    \label{thm:noapprox}
    \mdsp{} does not admit a~$k^{o(1)}$-approximation in~$f(k)\poly(n)$~time unless FPT~$=$~W[1].
    Assuming the gap-ETH, it cannot be~$o(k)$-approximated in~$f{(k)\poly(n)}$~time.
    All of these results hold even for subcubic graphs with terminals of degree one.
\end{theorem}

\begin{proof}
    We present a strict approximation-preserving reduction from \textsc{Multicolored Clique} to \mdsp{} such that the maximum degree is three and each terminal vertex has degree one.
    Moreover, the maximum number~$\opt$ of vertex-disjoint shortest paths between terminal pairs will be equal to the largest clique in the original instance.
    The theorem then follows from the fact that a~$f(k)\poly(n)$-time $k^{o(1)}$-approximation for \textsc{Clique} would imply that FPT~$=$~W[1]~\cite{ChenFLL24,SK22}, a~$f(k)\poly(n)$-time $o(k)$-approximation for \textsc{Clique} refutes the gap-ETH \cite{CCK+17}, and that the textbook reduction from \textsc{Clique} to \textsc{Multicolored Clique} only increases the number of vertices by a quadratic factor and does not change the size of a largest clique in the graph \cite{CyganFKLMPPS15}.

    The reduction is depicted in \cref{fig:hardness-example} and works as follows.
        \begin{figure}[t]
		\centering
		\begin{tikzpicture}
			\newcommand{\colorA}{blue}
			\newcommand{\colorB}{red!90!black}
			\newcommand{\colorC}{green!60!black}
			\newcommand{\colorD}{black}
			
			\newcommand{\smallDist}{0.2}
			\newcommand{\smallerDist}{0.05}
			\newcommand{\cliqueVertices}{1, 3, 2, 1} 
		
			\begin{scope}[xshift=5cm,yshift=-1cm,scale=0.7]
				\foreach \x in {1,...,4}{
					\node[vertex,fill=\colorA,
						label=above:{$1,\x$}
						] (v-1\x) at (\x,0) {};
					\node[vertex,fill=\colorB,
						label=right:{$2,\x$}
						] (v-2\x) at (5,-\x) {};
					\node[vertex,fill=\colorC,
						label=below:{$3,\x$}
						] (v-3\x) at (\x,-5) {};
					\node[vertex,fill=\colorD,
						label=left:{$4,\x$}
						] (v-4\x) at (0,-\x) {};
				}
			\end{scope}
			
			\begin{scope}[xshift=0cm, scale=2]
			\tikzset{vertex/.style ={circle,inner sep=1pt,minimum size=4pt}}
				\foreach[count=\i] \color in {\colorA,\colorB,\colorC,\colorD}{
					\pgfmathtruncatemacro\j{\i + 4};
					\node[vertex,fill=\color,label=below:{$s_\i$}] (s-\i) at (0,-\i) {};
    				\node[vertex,fill=\color,label=right:{$t_\i$}] (tt-\i) at (\i,-5) {};
					
					\foreach \j in {1,...,4}
					{
    					\pgfmathtruncatemacro\k{5 - \j};
						\node[vertex,fill=\color] (sc-\j\i) at (\smallDist,2.5 * \smallDist -\i - \j * \smallDist) {};
						\draw[dotted,thick] (sc-\j\i) -- (s-\i);
						\node[vertex,fill=\color] (ttc-\k\i) at (\i + \j * \smallDist - 2.5 * \smallDist,-4.7) {};
						\draw[dotted,thick] (ttc-\k\i) -- (tt-\i);
      
        				\begin{pgfonlayer}{background}
                            \node[vertex] (m1-\j\i) at (\i - 3 * \smallDist - \smallerDist,2.5 * \smallDist -\i - \j * \smallDist) {};
                            \node[vertex] (m2-\k\i) at (\i + \j * \smallDist - 2.5 * \smallDist, - 3 * \smallDist - \i - \smallerDist) {};
                        \end{pgfonlayer}
					}
				}
						
				\begin{pgfonlayer}{background}
                	\foreach[count=\i] \j in \cliqueVertices
					{
						\foreach[count=\x] \y in \cliqueVertices
						{
							\ifnum \i<\x%
								\draw[very thick] (v-\i\j) -- (v-\x\y);
							\fi
						}
						\draw[path] (s-\i.center) -- (sc-\j\i.center) -- (m1-\j\i.center) to[bend left=45] (m2-\j\i.center) -- (ttc-\j\i.center) -- (tt-\i.center);
					}
					\foreach \i in {1,...,4} \foreach \x in {\i,...,4} \foreach \j in {1,...,4} \foreach \y in {1,...,4}
                    {
                        \ifnum \x > \i
                            \xintifboolexpr { \i\j\x\y == 1333 || \i\j\x\y == 1233 || \i\j\x\y == 3143 || \i\j\x\y == 2331 || \i\j\x\y == 1341 || \i\j\x\y == 1223 || \i\j\x\y == 1143 || \i\j\x\y == 3342 || \i\j\x\y == 1222 || \i\j\x\y == 1231 || \i\j\x\y == 2143 || \i\j\x\y == 2242 || \i\j\x\y == 2331 || \i\j\x\y == 3242 || \i\j\x\y == 1123 || \i\j\x\y == 1132 || \i\j\x\y == 1141 || \i\j\x\y == 2332 || \i\j\x\y == 2341 || \i\j\x\y == 3241 || \i\j\x\y == 1424 || \i\j\x\y == 1434 || \i\j\x\y == 1444 || \i\j\x\y == 2434 || \i\j\x\y == 2444 || \i\j\x\y == 3444} 
                            {
                                \draw (v-\i\j) -- (v-\x\y);
                                \node[vertex,fill=black,minimum size=2] (a-\i\j\x\y) at (\i - \j * \smallDist + 2.5 * \smallDist,2.5 * \smallDist -\x - \y * \smallDist+\smallerDist) {};
                                \node[vertex,fill=black,minimum size=2] (b-\i\j\x\y) at (\i - \j * \smallDist + 2.5 * \smallDist,2.5 * \smallDist -\x - \y * \smallDist-\smallerDist) {};
                                \node[vertex,fill=black,minimum size=2] (c-\i\j\x\y) at (\i - \j * \smallDist + 2.5 * \smallDist-\smallerDist,2.5 * \smallDist -\x - \y * \smallDist) {};
                                \node[vertex,fill=black,minimum size=2] (d-\i\j\x\y) at (\i - \j * \smallDist + 2.5 * \smallDist+\smallerDist,2.5 * \smallDist -\x - \y * \smallDist) {};
                            }{
                                \node[vertex,fill=black,minimum size=2] (a-\i\j\x\y) at (\i - \j * \smallDist + 2.5 * \smallDist-.6*\smallerDist,2.5 * \smallDist -\x - \y * \smallDist+.6*\smallerDist) {};
                                \node[vertex,fill=black,minimum size=2] (b-\i\j\x\y) at (\i - \j * \smallDist + 2.5 * \smallDist+.6*\smallerDist,2.5 * \smallDist -\x - \y * \smallDist-.6*\smallerDist) {};
                                \node[vertex,fill=black,minimum size=2] (c-\i\j\x\y) at (\i - \j * \smallDist + 2.5 * \smallDist-.6*\smallerDist,2.5 * \smallDist -\x - \y * \smallDist+.6*\smallerDist) {};
                                \node[vertex,fill=black,minimum size=2] (d-\i\j\x\y) at (\i - \j * \smallDist + 2.5 * \smallDist+.6*\smallerDist,2.5 * \smallDist -\x - \y * \smallDist-.6*\smallerDist) {};
                            }
                        \fi
                    }
                    \node[label={\footnotesize{$s_1^1$}}] at (.3,-.8) {};
                    \node[label={\footnotesize{$s_1^4$}}] at (.3,-1.7) {};
                    \node[label={\footnotesize{$t_1^1$}}] at (1.45,-4.9) {};
                    \node[label={\footnotesize{$t_1^4$}}] at (.55,-4.9) {};
                    \node[label={\footnotesize{$u^{1,2}_{1,1}$}}] at (1.3,-1.8) {};
                    \node[label={\footnotesize{$v^{1,2}_{1,1}$}}] at (1.55,-2) {};
                    \node[label={\footnotesize{$u^{2,3}_{3,1}$}}] at (1.85,-2.75) {};
                    \node[label={\footnotesize{$x^{1,3}_{4,4}$}}] at (.47,-3.63) {};

                    \foreach \i in {1,...,4} \foreach \x in {\i,...,4} \foreach \j in {1,...,4} \foreach \y in {1,...,4}
                    {
                        \ifnum \x > \i
                            \pgfmathtruncatemacro\k{\j+1};
                            \pgfmathtruncatemacro\l{\y+1};
                            \draw (a-\i\j\x\y) -- (b-\i\j\x\y);
                            \draw (c-\i\j\x\y) -- (d-\i\j\x\y);
                            \ifnum \j < 4
                                \draw[blue] (c-\i\j\x\y) -- (d-\i\k\x\y); 
                            \fi
                            \ifnum \y < 4
                                \draw[blue] (b-\i\j\x\y) -- (a-\i\j\x\l); 
                            \fi
                        \fi
                    }

                    \foreach \i in {1,...,3} \foreach \x in {\i,...,4} \foreach \j in {1,...,4} \foreach \y in {4} \foreach \z in {1}
                    {
                        \pgfmathtruncatemacro\h{\i+1};
                        \ifnum \x > \h
                            \draw[blue] (d-\i\z\x\j) -- (c-\h\y\x\j);
                        \fi
                        \ifnum \x > \i
                            \ifnum \x < 4
                                \pgfmathtruncatemacro\h{\x+1};
                                \draw[blue] (b-\i\j\x\y) -- (a-\i\j\h\z);
                            \fi
                        \fi
                    }

                    \foreach \i in {2,...,4} \foreach \j in {1,...,4} \foreach \y in {4} \foreach \z in {1}
                    {
                        \draw[red] (sc-\j\i) -- (c-\z\y\i\j);
                    }
                    \foreach \i in {1,...,3} \foreach \j in {1,...,4} \foreach \y in {4} \foreach \z in {1}
                    {
                        \draw[red] (b-\i\j\y\y) -- (ttc-\j\i);
                    }
                    \foreach \i in {2,3} \foreach \j in {1,...,4} \foreach \y in {4} \foreach \z in {1}
                    {
                        \pgfmathtruncatemacro\f{\i-1};
                        \pgfmathtruncatemacro\g{\i+1};
                        \pgfmathtruncatemacro\k{5-\j};
                        \draw[blue] (d-\f\z\i\j) edge[bend left=45] (a-\i\j\g\z);
                    }
                    \foreach \j in {1,...,4} \foreach \y in {4} \foreach \z in {1} \foreach \a in {2} \foreach \b in {3}
                    {
                        \draw[red] (sc-\j\z) edge[bend left=45] (a-\z\j\a\z);
                        \draw[red] (d-\b\z\y\j) edge[bend left=45] (ttc-\j\y);
                    }

				\end{pgfonlayer}
			\end{scope}
		\end{tikzpicture}
		\caption{
			An illustration of the reduction from \textsc{Multicolored Clique} to \mdsp. \\
			\emph{Top right:} Example instance for \textsc{Multicolored Clique} with~$k=4$ colors and~$n=4$ vertices per color.
			A multicolored clique is highlighted (by thick edges). \\
			\emph{Bottom left:} The constructed instance with the four shortest paths corresponding to the vertices of the clique highlighted.
			Note that these paths are pairwise disjoint.
			The dotted edges (incident to~$s_i$ and~$t_i$ vertices) indicate binary trees (where all leaves have distance~$\lceil\log \nu \rceil$ from the root).
            Red edges indicate paths of length~$2\nu$ and blue edges indicate paths of length~$2$.
		}
		\label{fig:hardness-example}
	\end{figure}%
    Let~$G=(V,E)$ be a~$k$-partite graph (or equivalently a graph colored with~$k$ colors where all vertices of any color form an independent set) with~$\nu$ vertices of each color.
    Let~$V_i=\{v^i_1,v^i_2,\ldots,v^i_\nu\}$ be the set of vertices of color~$i \in [k]$ in~$G$.
    We start with a terminal pair~$(s_i,t_i)$ for each color~$i$ and a pair of (non-terminal) vertices~$(s_i^j,t_i^j)$ for each vertex~$v^i_j \in V_i$.
    Next for each color~$i$, we add a binary tree of height~$\lceil \log \nu \rceil$ where the vertices~$s_i^j$ are the leaves for all~$v_j^i \in V_i$.
    We make~$s_i$ adjacent to the root of the binary tree.
    Analogously, we add a binary tree of the same height with leaves~$t^i_j$ and make~$t_i$ adjacent to the root.
    Next, we add a crossing gadget for each pair of vertices~$(v_j^i,v_b^a)$ with~$i < a$.
    If~$\{v_j^i,v_b^a\} \in E$, then the gadget consists of four vertices~$u_{j,b}^{i,a},v_{j,b}^{i,a},x_{j,b}^{i,a}$, and~$y_{j,b}^{i,a}$ and edges~$\{u_{j,b}^{i,a},v_{j,b}^{i,a}\}$ and~$\{x_{j,b}^{i,a},y_{j,b}^{i,a}\}$.
    If~$\{v_j^i,v_b^a\} \notin E$, then the gadget consists of only two vertices~$u_{j,b}^{i,a}$ and~$v_{j,b}^{i,a}$ and the edge~$\{u_{j,b}^{i,a},v_{j,b}^{i,a}\}$.
    For the sake of notational convenience, we will in the latter case also denote~$u_{j,b}^{i,a}$ by~$x_{j,b}^{i,a}$ and~$v_{j,b}^{i,a}$ by~$y_{j,b}^{i,a}$.
    To complete the construction, we connect the different gadgets as follows.
    First, we connect via paths of length two~$v_{j,b}^{i,a}$ and~$u_{j,b+1}^{i,a}$ for all~$b<\nu$ and~$y_{j,b}^{i,a}$ and~$x_{j-1,b}^{i,a}$ for all~$j > 1$.
    Second, we connect via paths of length two the vertices~$v_{j,\nu}^{i,a}$ to~$u_{j,1}^{i,a+1}$ for all~$j \in [\nu]$ and all~$a < k$ and~$y^{i,a}_{1,b}$ to~$x^{i+1,a}_{\nu,b}$ for all~$b \in [\nu]$ and all~$i < a-1$.
    Third, we connect also via paths of length two~$y^{i,i+1}_{1,b}$ to~$u^{i+1,i+2}_{b,1}$ for all~$i < k-1$ and all~$b \in [\nu]$.
    Next, we connect via paths of length~$2\nu$ each vertex~$s_i^j$ to~$x^{1,i}_{\nu,j}$ for each~$i > 1$ and~$j \in [\nu]$.
    Similarly, $v^{i,k}_{j,\nu}$ is connected to~$t_i^j$ via paths of length~$2\nu$.
    Finally, we connect~$s^1_j$ with~$u^{1,2}_{j,1}$ for all~$j \in [\nu]$ and~$y^{k-2,k-1}_{1,j}$ with~$t^k_j$ for all~$j \in [\nu]$.
    This concludes the construction.

    We next prove that all shortest~$s_i$-$t_i$-paths are of the form
    \begin{align}
       s_i - s_i^j - x^{1,i}_{\nu,j} - y^{i-1,i}_{1,j} - u^{i,i+1}_{j,1} - v^{i,k}_{j,\nu} - t^j_i - t_i \label{eq:canonical}
    \end{align}
    for some~$j \in [\nu]$ and where the~$s_1$-$t_1$-paths go directly from~$s_i^j$ to~$u^{1,2}_{j,1}$ and the~$s_k$-$t_k$-paths go directly from~$y^{k-1,k}_{1,j}$ to~$t_j^i$.
    We say that the respective path is the \emph{$j$\textsuperscript{th} canonical path} for color~$i$.

    To show the above claim, first note that the distance from~$s_i$ to any vertex~$s^j_{i}$ is the same value~${x = \lceil \log \nu \rceil+1}$ for all pairs of indices~$i$ and~$j$.
    Moreover, the same holds for~$t_i$ and~$t_{i}^{j}$, each~$s_i$-$t_i$-path contains at least one vertex~$s_i^j$ and one vertex~$t_i^{j'}$ for some~$j,j' \in [\nu]$, and all paths of the form in \cref{eq:canonical} are of length~$y=2x + 4 \nu + 3 (k-1) \nu - 2$.
    We first show that each~$s_i$-$t_i$-path of length at most~$y$ contains an edge of the form~$y^{i,i+1}_{1,b}$ to~$u^{i+1,i+2}_{b,1}$.
    Consider the graph where all of these edges are removed.
    Note that due to the grid-like structure, the distance between~$s_i$ and~$x^{i',a}_{j,b}$ for any values~$i' \leq i \leq a,j$, and~$b$ is at least~$x+2\nu+3(i'-1)\nu+3(\nu-j)$ if~$i = a$ and at least~$x+2\nu+3(i'-1)\nu+3(a-i)\nu+3(\nu-j)+3b$ if~$i < a$.\footnote{We mention that there are some pairs of vertices~$x^{i_1,a_1}_{j_1,b_1}$ and~$x^{i_2,a_2}_{j_2,b_2}$, where the distance between the two is less than~$3(|i_1-i_2|+|a_1-a_2|)\nu+3(|j_1-j_2|+|b_1-b_2|)$. An example is the pair~$(x^{1,2}_{1,1}=u^{1,2}_{1,1},x^{1,2}_{2,2})$~in \cref{fig:hardness-example} with a distance of~$4$. However, in all examples it holds that~$i_1\nu-j_1 \neq i_2\nu-j_2$ and~$a_1\nu+b_1 \neq a_2\nu+b_2$ such that the left side is either smaller in both inequalities or larger in both inequalities. Hence, these pairs cannot be used as shortcuts as they move ``down and left'' instead of ``down and right'' in \cref{fig:hardness-example}.}
    Hence, all shortest~$s_i$-$t_i$-paths use an edge of the form~$y^{i,i+1}_{1,b}$ to~$u^{i+1,i+2}_{b,1}$ and the shortest path from~$s_i^j$ to some vertex~$y^{i,i+1}_{1,b}$ is to the vertex~$y^{i,i+1}_{1,j}$.
    Note that the other endpoint of the specified edge is~$u^{i,i+1}_{j,1}$ and the shortest path to~$t_i$ now goes via~$t_i^j$ for analogous reasons.
    Thus, all shortest~$s_i$-$t_i$-paths have the form~(\ref{eq:canonical}).
    
    We next prove that any set of~$p$ disjoint shortest paths between terminal pairs~$(s_i,t_i)$ in the constructed graph has a one-to-one correspondence to a multicolored clique of size~$p$ for any~$p$.
    For the first direction, assume that there is a set~$P$ of disjoint shortest paths between~$p$ terminal pairs.
    Let~$S \subseteq [k]$ be the set of indices such that the paths in~$P$ connect~$s_i$ and~$t_i$ for each~$i \in S$.
    Moreover, let~$j_i$ be the index such that the shortest~$s_i$-$t_i$-path in~$P$ is the~$j_i$\textsuperscript{th} canonical path for~$i$ for each~$i \in S$.
    Now consider the set~$K = \{v_{j_i}^i \mid i \in S\}$ of vertices in~$G$.
    Clearly~$K$ contains at most one vertex of each color and is of size~$p$ as~$S$ is of size~$p$.
    It remains to show that~$K$ induces a clique in~$G$.
    To this end, consider any two vertices~$v_{j_i}^i, v_{j_{i'}}^{i'} \in K$.
    We assume without loss of generality that~$i < i'$.
    By assumption, the~$j_i$\textsuperscript{th} canonical path for~$i$ and the~$j_{i'}$\textsuperscript{th} canonical path for~$i'$ are disjoint.
    This implies that~$u^{i,i'}_{j_i,j_{i'}} \neq x^{i,i'}_{j_i,j_{i'}}$ as the~$j_i$\textsuperscript{th} canonical path for~$i$ contains the former and the~$j_{i'}$\textsuperscript{th} canonical path for~$i'$ contains the latter.
    By construction, this means that~$\{v_{j_i}^i, v_{j_{i'}}^{i'}\} \in E$.
    Since the two vertices were chosen arbitrarily, it follows that all vertices in~$K$ are pairwise adjacent, that is,~$K$ induces a multicolored clique of size~$p$.

    For the other direction assume that there is a multicolored clique~$C = \{v_{j_1}^{i_1},v_{j_2}^{i_2},\ldots,v_{j_p}^{i_p}\}$ of size~$p$ in~$G$.
    We will show that the~$j_q$\textsuperscript{th} canonical path for~$i_q$ is vertex-disjoint from the~$j_r$\textsuperscript{th} canonical path for~$i_r$ for all~$q \neq r \in [p]$.
    Let~$q,r$ be two arbitrary distinct indices in~$[p]$ and let without loss of generality be~$q < r$.
    Note that the two mentioned paths can only overlap in vertices~$u^{i_q,i_r}_{j_q,j_r},v^{i_q,i_r}_{j_q,j_r},x^{i_q,i_r}_{j_q,j_r},$ or~$y^{i_q,i_r}_{j_q,j_r}$ and that the~$j_q$\textsuperscript{th} canonical path for~$i_q$ only contains vertices~$u^{i_q,i_r}_{j_q,j_r}$ and~$v^{i_q,i_r}_{j_q,j_r}$ and the~$j_r$\textsuperscript{th} canonical path for~$i_r$ only contains~$x^{i_q,i_r}_{j_q,j_r}$ and~$y^{i_q,i_r}_{j_q,j_r}$.
    Moreover, since by assumption~$v_{j_q}^{i_q}$ and~$v_{j_r}^{i_r}$ are adjacent, it holds by construction that~$u^{i_q,i_r}_{j_q,j_r}, v^{i_q,i_r}_{j_q,j_r}, x^{i_q,i_r}_{j_q,j_r}$, and~$y^{i_q,i_r}_{j_q,j_r}$ are four distinct vertices.
    Thus, we found vertex disjoint paths between~$p$ distinct terminal pairs.
    This concludes the proof of correctness.

    To finish the proof, observe that the constructed instance has maximum degree three, all terminal vertices have degree one, and the construction can be computed in polynomial time.
\end{proof}

We mention in passing that in graphs of maximum degree three with terminal vertices of degree one, two paths are vertex disjoint if and only if they are edge disjoint.
Hence, \cref{thm:noapprox} also holds for \textsc{Maximum Edge-Disjoint Shortest Paths}, the edge-disjoint version of \mdsp.

\begin{corollary}
    \label{cor:noapprox}
    \textsc{Maximum Edge-Disjoint Shortest Paths} does not admit a~$k^{o(1)}$-approximation in~$f(k)\poly(n)$~time unless FPT~$=$~W[1].
    Assuming the gap-ETH, it cannot be~$o(k)$-approximated in~$f{(k)\poly(n)}$~time.
    All of these results hold even for subcubic graphs with terminals of degree one.
\end{corollary}

We continue with an unparameterized lower bound by establishing that computing a~$m^{\frac{1}{2}-\varepsilon}$-approxi\-mation is NP-hard.
We mention that the reduction is quite similar to the reduction in the proof for \cref{thm:noapprox}. 

\begin{restatable}{theorem}{nosqrt} 
    \label{thm:nosqrt}
    Computing a~$m^{\nicefrac{1}{2}-\varepsilon}$-approximation for any $\varepsilon > 0$ for \mdsp{} is NP-hard.
\end{restatable}

{

\begin{proof}
    It is known that computing a~$n^{1-\varepsilon}$-approximation for \textsc{Clique} is NP-hard~\cite{Has96,Zuc07}.
    We present an approximation-preserving reduction from \textsc{Clique} to \mdsp{} based on \cref{thm:noapprox}.
    We use basically the same reduction as in \cref{thm:noapprox} but we start from an instance of \textsc{Clique} and have a separate terminal pair for each vertex in the graph.
    Moreover, we do not require the binary trees pending from the terminal vertices and neither do we require long induced paths (red edges in \cref{fig:hardness-example}).
    These are instead paths with one internal vertex.
    An illustration of the modified reduction is given in \cref{fig:hardness-example2}.
    \begin{figure}[t]
		\centering
		\begin{tikzpicture}
			\newcommand{\colorD}{black}
			
			\newcommand{\smallDist}{0.2}
			\newcommand{\smallerDist}{0.05}
		
			\begin{scope}[yshift=-1cm,yshift=-2.8cm]
                \node[circle,draw,label=$v_1$] (v1) at (0,2) {};
                \node[circle,draw,label=$v_2$] (v2) at (2,2) {} edge[very thick] (v1);
                \node[circle,draw,label=$v_3$] (v3) at (4,2) {} edge (v2);
                \node[circle,draw,label=below:$v_4$] (v4) at (0,0) {} edge (v1) edge (v3);
                \node[circle,draw,label=below:$v_5$] (v5) at (2,0) {} edge[very thick] (v1) edge[very thick] (v2) edge (v4);
                \node[circle,draw,label=below:$v_6$] (v6) at (4,0) {} edge[very thick] (v1) edge[very thick] (v2) edge (v3) edge[very thick] (v5);
            \end{scope}
			
			\begin{scope}[xshift=6.5cm, xscale=1.75, yscale=1.5]
			
				\foreach \i in {1,...,6}{
					\node[vertex,fill=\colorD,label=left:{$s_\i$}] (s-\i) at (0,-.5*\i) {};
    				\node[vertex,fill=\colorD,label=below:{$t_\i$}] (t-\i) at (.5*\i,-3.5) {};
    				\node (m-\i) at (.5*\i-.5,-.5*\i) {};
    				\node (w-\i) at (.5*\i,-.5*\i-.5) {};
				}
						
				\begin{pgfonlayer}{background}
                    \draw[path] (s-1.center) -- (m-1.center) -- (w-1.center) -- (t-1.center);
                    \draw[path] (s-2.center) -- (m-2.center) -- (w-2.center) -- (t-2.center);
                    \draw[path] (s-5.center) -- (m-5.center) -- (w-5.center) -- (t-5.center);
                    \draw[path] (s-6.center) -- (m-6.center) -- (w-6.center) -- (t-6.center);
                    
					\foreach \i in {1,...,6} \foreach \x in {\i,...,6}
                    {
                        \ifnum \x > \i
                            \xintifboolexpr {\i\x == 12 || \i\x == 14 || \i\x == 15 || \i\x == 16 || \i\x == 23 || \i\x == 25 || \i\x == 26 || \i\x == 34 || \i\x == 36 || \i\x == 45 || \i\x == 56} 
                            {
                                \node[vertex,fill=black,minimum size=2] (a-\i\x) at (0.5*\i,-0.5*\x + \smallerDist) {};
                                \node[vertex,fill=black,minimum size=2] (b-\i\x) at (0.5*\i,-0.5*\x - \smallerDist) {};
                                \node[vertex,fill=black,minimum size=2] (c-\i\x) at (0.5*\i - \smallerDist,-0.5*\x) {};
                                \node[vertex,fill=black,minimum size=2] (d-\i\x) at (0.5*\i + \smallerDist,-0.5*\x) {};
                            }{
                                \node[vertex,fill=black,minimum size=2] (a-\i\x) at (0.5*\i - \smallerDist,-0.5*\x + \smallerDist) {};
                                \node[vertex,fill=black,minimum size=2] (b-\i\x) at (0.5*\i + \smallerDist,-0.5*\x - \smallerDist) {};
                                \node[vertex,fill=black,minimum size=2] (c-\i\x) at (0.5*\i - \smallerDist,-0.5*\x + \smallerDist) {};
                                \node[vertex,fill=black,minimum size=2] (d-\i\x) at (0.5*\i + \smallerDist,-0.5*\x - \smallerDist) {};
                            }
                        \fi
                    }

                    \foreach \i in {1,...,6} \foreach \x in {\i,...,6}
                    {
                        \ifnum \x > \i
                            \draw (a-\i\x) -- (b-\i\x);
                            \draw (c-\i\x) -- (d-\i\x);
                        \fi
                    }

                    \foreach \i in {1,...,5} \foreach \x in {\i,...,6}
                    {
                        \pgfmathtruncatemacro\h{\i+1};
                        \ifnum \x > \h
                            \draw (d-\i\x) -- (c-\h\x) node[midway, vertex,fill=black,minimum size=2] {};
                        \fi
                        \ifnum \x > \i
                            \ifnum \x < 6
                                \pgfmathtruncatemacro\h{\x+1};
                                \draw (b-\i\x) -- (a-\i\h) node[midway, vertex,fill=black,minimum size=2] {};
                            \fi
                        \fi
                    }

                    \foreach \i in {2,...,6} \foreach \z in {1}
                    {
                        \draw (s-\i) -- (c-\z\i) node[midway, vertex,fill=black,minimum size=2] {};
                    }
                    \foreach \i in {1,...,5} \foreach \y in {6}
                    {
                        \draw (b-\i\y) -- (t-\i) node[midway, vertex,fill=black,minimum size=2] {};
                    }
                    \foreach \i in {2,...,5}
                    {
                        \pgfmathtruncatemacro\f{\i-1};
                        \pgfmathtruncatemacro\g{\i+1};
                        \draw (d-\f\i) -- (a-\i\g) node[midway, vertex,fill=black,minimum size=2] {};
                    }
                    \draw (s-1) -- (a-12) node[midway, vertex,fill=black,minimum size=2] {};
                    \draw (d-56) -- (t-6) node[midway, vertex,fill=black,minimum size=2] {};

				\end{pgfonlayer}
			\end{scope}
		\end{tikzpicture}
		\caption{
			An illustration of the reduction from \textsc{Clique} to \mdsp. \\
			\emph{Left side:} Example instance for \textsc{Clique} with a highlighted solution (by thick edges).\\
			\emph{Right side:} The constructed instance with the four shortest paths corresponding to the solution on the left side highlighted. Note that each shortest~$s_i$-$t_i$-path uses exactly two of the diagonal edges.
		}
		\label{fig:hardness-example2}
	\end{figure}%
    Note that the number of vertices and edges in the graph is at most~$3N^2$, where~$N$ is the number of vertices in the instance of \textsc{Clique}.
    Moreover, for each terminal pair~$(s_i,t_i)$, there is exactly one shortest~$s_i$-$t_i$-path (the path that moves horizontally in \cref{fig:hardness-example2} until it reaches the main diagonal, then uses exactly two edges on the diagonal, and finally moves vertically to~$t_i$).
    
    We next prove that for any~$p$, there is a one-to-one correspondence between a set of~$p$ disjoint shortest paths between terminal pairs~$(s_i,t_i)$ in the constructed graph and a clique of size~$p$ in the input graph.
    For the first direction, assume that there is a set~$P$ of disjoint shortest paths between~$p$ terminal pairs.
    Let~$S \subseteq [k]$ be the set of indices such that the paths in~$P$ connect~$s_i$ and~$t_i$ for each~$i \in S$.
    Now consider the set~$K = \{v_{i} \mid i \in S\}$ of vertices in~$G$.
    Clearly~$K$ contains~$p$ vertices.
    It remains to show that~$K$ induces a clique in~$G$.
    To this end, consider any two vertices~$v_{i}, v_{j} \in K$.
    We assume without loss of generality that~$i < j$.
    By assumption, the unique shortest~$s_i$-$t_i$-path and the unique shortest~$s_j$-$t_j$-path are vertex-disjoint.
    By the description of the shortest paths between terminal pairs and the fact that~$s_i$ is higher than~$s_j$ and~$t_i$ is to the left of~$t_j$, it holds that the two considered paths both visit the region that is to the right of~$s_j$ and above~$t_i$. 
    This implies that two edges must be crossing at this position, that is, there are four vertices in the described region and not only two.
    By construction, this means that~$\{v_i,v_j\} \in E$.
    Since the two vertices were chosen arbitrarily, it follows that all vertices in~$K$ are pairwise adjacent, that is,~$K$ induces a clique of size~$p$ in the input graph.

    For the other direction assume that there is a clique~$C = \{v_{i_1},v_{i_2},\ldots,v_{i_p}\}$ of size~$p$ in the input graph.
    We will show that the unique shortest~$s_{i_q}$-$t_{i_q}$-path is vertex-disjoint from the unique shortest~\mbox{$s_{i_r}$-$t_{i_r}$-path} for all~$q \neq r \in [p]$.
    Let~$q,r$ be two arbitrary distinct indices in~$[p]$ and let without loss of generality be~$q < r$.
    Note that the two mentioned paths can only overlap in the region that is to the right of~$s_{i_r}$ and above~$t_{i_q}$.
    Moreover, since by assumption~$v_{i_q}$ and~$v_{i_r}$ are adjacent, it holds by construction that there are four distinct vertices in the described region and the two described paths are indeed vertex-disjoint.
    Thus, we found vertex disjoint paths between~$p$ distinct terminal pairs.
    
    We conclude by analyzing the approximation ratio.
    Note that we technically did not prove a strict reduction for the factor~$m^{1-\varepsilon}$ as the number of vertices in the original instance and the number of edges in the constructed instance are not identical.
    Still, the number~$m$ of edges in the constructed instance is at most~$3N^2$, where~$N$ is the number of vertices in the original instance of \textsc{Clique}.
    Hence, any~$m^{\nicefrac{1}{2}-\varepsilon}$-approximation for \mdsp{} corresponds to a~$(3N^2)^{\nicefrac{1}{2}-\varepsilon} = N^{1-\varepsilon'}$-approximation for \textsc{Clique} for some~$0 < \varepsilon' < 2\varepsilon$ and therefore computing a~$m^{\nicefrac{1}{2}-\varepsilon}$-approximation for \mdsp{} is NP-hard.
\end{proof}

Note that the maximum degree of the constructed instance is again three and all terminal vertices are of degree one.
Thus, \cref{thm:nosqrt} also holds for the edge disjoint version of \mdsp.
However in this case, a very similar result was already known before.
Guruswami et al.~\cite{GKRSY03} showed that computing a~$m^{\nicefrac{1}{2}-\varepsilon}$-approximation is NP-hard for a related problem called \textsc{Length Bounded Edge-Disjoint Paths}.
Their reduction immediately implies the same hardness for \textsc{Maximum Edge-Disjoint Shortest Paths}.
To the best of our knowledge, the best known unparameterized approximation lower bound for \mdsp{} was the $2-\varepsilon$ lower bound due to Chitnis et al.~\cite{CTW24} and we are not aware of any lower bound for \mdp.}

We next show that \cref{thm:nosqrt} is tight, that is, we show how to compute a~$\sqrt{n}$-approximation for \mdsp{} in polynomial time.
We also show that the same algorithm achieves a~$\lceil\sqrt{\ell}\rceil$-approximation.
Note that we can always assume that~$\ell \leq \min(n,m)$ as a set of vertex-disjoint paths is a forest and the number of edges in a forest is less than its number of vertices.
We mention that this algorithm is basically identical to the best known (unparameterized) approximation algorithm for \textsc{Maximum Disjoint Paths}~\cite{Klein96,KS98}.

\begin{theorem}
    \label{prop:lapprox}
    There is a polynomial-time algorithm for \mdsp{} on directed and weighted graphs that achieves an approximation factor of~$\min\{\sqrt{n},\lceil\sqrt{\ell}\rceil\}$.
\end{theorem}
\begin{proof}
    Let~$\opt$ be a maximum subset of terminal pairs that can be connected by shortest pairwise vertex-disjoint paths and let~$j$ be the index of a terminal pair~$(s_j,t_j)$ such that a shortest~$(s_j,t_j)$-path contains a minimum number of arcs.
    We can compute the index~$j$ as well as a shortest~$s_j$-$t_j$-path with a minimum number of arcs by running a folklore modification of Dijkstra's algorithm from each terminal vertex~$s_i$.\footnote{The standard Dijkstra's algorithm is modified by assigning to each vertex a pair of labels: the distance from the terminal and the number of arcs in the corresponding path; then the pairs of labels are compared lexicographically.}
    Let~$\ell_j$ be the number of arcs in the found path.
    Our algorithm iteratively picks the shortest~$s_j$-$t_j$-path using~$\ell_j$ arcs, removes all involved vertices from the graph, recomputes the distance between all terminal pairs, removes all terminal pairs whose distance increased, updates the index~$j$, and recomputes~$\ell_j$.
    We distinguish whether~$\ell_j + 1 \leq \min(\sqrt{n},\lceil\sqrt{\ell}\rceil)$ or not.
    
    While~$\ell_j + 1 \leq \min(\sqrt{n},\lceil\sqrt{\ell}\rceil)$, note that we removed at most~$\ell_j+1$ terminals pairs in~$\opt$.
    Hence, if~$\ell_j + 1 \leq \min(\sqrt{n},\lceil\sqrt{\ell}\rceil)$ holds at every stage, then we connected at least~$\nicefrac{|\opt|}{\min(\sqrt{n},\lceil\sqrt{\ell}\rceil)}$ terminal pairs, that is, we found a~$\min(\sqrt{n},\lceil\sqrt{\ell}\rceil)$-approximation.
    
    So assume that at some point~$\ell_j + 1 > \min(\sqrt{n},\lceil\sqrt{\ell}\rceil)$ and let~$x$ be the number of terminal pairs that we already connected by disjoint shortest paths.
    By the argument above, we have removed at most~$x \cdot \min(\sqrt{n},\lceil\sqrt{\ell}\rceil)$ terminal pairs from~$\opt$ thus far.
    We now make a case distinction whether or not~$\sqrt{n} \leq \lceil\sqrt{\ell}\rceil$.
    If~$\ell_j + 1 > \lceil\sqrt{\ell}\rceil \geq \sqrt{n}$, then we note that all remaining paths in~$\opt$ contain at least~$\sqrt{n}$ vertices each and since the paths are vertex-disjoint, there can be at most~$\sqrt{n}$ paths left in~$\opt$.
    Hence, we can infer that~${|\opt| \leq (x+1) \cdot \sqrt{n}}$.
    Consequently, even though we might remove all remaining terminal pairs in~$\opt$ by connecting~$s_j$ and~$t_j$, this is still a~$\sqrt{n}$-approximation (and a~$\lceil\sqrt{\ell}\rceil$-approximation as we assumed~$\lceil\sqrt{\ell}\rceil \geq \sqrt{n}$).
    
    If~$\ell_j + 1> \sqrt{n} \geq \lceil\sqrt{\ell}\rceil$, then we note that all remaining paths in~$\opt$ contain at least~$\ell_j > \lceil\sqrt{\ell}\rceil-1$ edges each.
    Moreover, since~$\ell_j$ and~$\lceil\sqrt{\ell}\rceil$ are integers, each path contains at least~$\lceil\sqrt{\ell}\rceil$ edges each.
    Since all paths in~$\opt$ contain by definition at most~$\ell$ edges combined, the number of paths in~$\opt$ is at most~$\nicefrac{\ell}{\lceil\sqrt{\ell}\rceil} \leq \lceil\sqrt{\ell}\rceil$.
    Hence, we can infer in that case that~$|\opt| \leq (x+1) \cdot \lceil\sqrt{\ell}\rceil$.
    Again, even if we remove all remaining terminal pairs in~$\opt$ by connecting~$s_j$ and~$t_j$, this is still a~$\lceil\sqrt{\ell}\rceil$-approximation (and a~$\sqrt{n}$-approximation as we assumed~$\sqrt{n} \geq \lceil\sqrt{\ell}\rceil$).
    This concludes the proof.
\end{proof}

\section{Exact Algorithms}
\label{sec:fpt}
In this section, we present a~$2^{O(\ell)}\poly(n)$-time algorithm for \mdsp, that is, we show that the problem is fixed-parameter tractable when parameterized by~$\ell$.
On the negative side, we show that even the special case \dsp{} does not admit a polynomial kernel when parameterized by~$\ell$ and cannot be solved in~$2^{o(n+m)}$ unless the ETH fails.
Note that since~$\ell \leq \min(n,m)$, this also excludes~$2^{o(\ell)}\poly(n)$-time algorithms under the ETH.
We mention that our reduction also excludes~$2^{o(\sqrt{n})}$-time algorithms for \dsp{} on planar input graphs.

The algorithmic result uses the technique of \emph{color coding} of Alon, Yuster, and Zwick~\cite{AYZ95}.
Imagine we are searching for some structure of size~$k$ in a graph.
The idea of color coding is to color the vertices (or edges) of the input graph with a set of~$k$~colors and then only search for colorful solutions, that is, structures in which all vertices have distinct colors.
Of course, this might not yield an optimal solution, but by trying enough different random colorings, one can often get a constant error probability in~$f(k)\poly(n)$ time.
Using the following result by Naor, Schulman, and Srinivasan~\cite{NSS95}, this can also be turned into a deterministic algorithm showing that the problem is fixed-parameter tractable.
The result states that for any~$n,k \geq 1$, one can construct an~$(n,k)$-perfect hash family of size~$e^kk^{O(\log k)}\log n$ in~$e^kk^{O(\log k)}n\log n$ time.
An~$(n,k)$-perfect hash family~$\mathcal{F}$ is a family of functions from~$[n]$ to~$[k]$ such that for every set~$S \subseteq [n]$ with~$|S| \leq k$, there exists a function~$f \in \mathcal{F}$ such that~$f$ colors all vertices in~$S$ with distinct colors.

\begin{restatable}{theorem}{fptl} 
    \label{prop:fptl}
    \mdsp{} on weighted and directed graphs can be solved in~$2^{O(\ell)}\poly(n)$~time.
\end{restatable}

{

\begin{proof}
  Let~$(G,w,(s_1,t_1),\dots,(s_k,t_k),p)$ be an instance of \mdsp.
  First, we guess the value of~$\ell$ by starting with~$\ell=p$ and increasing the value of~$\ell$ by one whenever we cannot find a solution with at least~$p$ shortest paths and at most~$\ell$ edges.
  We start with~$\ell = p$ as any set of~$p$ disjoint paths contains at least~$\ell$ edges.
  Notice that the total number of vertices in a (potential) solution with $p$ paths is at most~$\ell+p$.
  We use the color-coding technique of Alon, Yuster, and Zwick~\cite{AYZ95}.
  We color the vertices of~$G$ uniformly at random using~$p+\ell$~colors (the set of colors is~$[\ell+p]$) and observe that the probability that all the vertices in the paths in a solution have distinct colors is at least~$\frac{(p+\ell)!}{(p+\ell)^{(p+\ell)}}\geq e^{-(p+\ell)}$. 
  We say that a solution to the considered instance is \emph{colorful} if distinct paths in the solution have no vertices of the same color.
  Note that we do not require that the vertices within a path in the solution are colored by distinct/equal colors.
  The crucial observations are that any colorful solution is a solution and the probability of the existence of a colorful solution for a yes-instance of \mdsp{} is at least~$e^{-(p+\ell)}$ as any solution in which all vertices receive distinct colors is a colorful solution. 
  
  We use dynamic programming over subsets of colors to find a colorful solution. More precisely, we find the minimum number of arcs in a
  collection~${\mathcal{C}=\{P_i\}_{i\in S}}$ of $p$ pairwise vertex-disjoint paths for some $S\subseteq [k]$ satisfying the conditions: 
  (i) for each~$i \in S$, the path~$P_i$ is a shortest path from~$s_i$ to~$t_i$ and (ii) there are no vertices of distinct paths of the same color.

  For a subset~$X\subseteq [p+\ell]$ of colors and a positive integer~$r\leq p$, we denote by~$f[X,r]$ the minimum total number of arcs in~$r$~shortest paths connecting distinct terminal pairs such that the paths contains only vertices of colors in~$X$ and there are no vertices of distinct paths of the same color.
  We set~${f[X,r]= \infty}$ if such a collection of~$r$~paths does not exist.

  To compute~$f$, if~$r=1$, then let~$W\subseteq V$ be the subset of vertices colored by the colors in~$X$.
  We use Dijkstra's algorithm to find the set~$I\subseteq[k]$ of all indices~$i\in[k]$ such that the lengths of the shortest~$s_i$-$t_i$-paths in $G$ and $G[W]$ are the same.
  If $I=\emptyset$, then we set~$f[X,1]=\infty$.
  Assume that this is not the case.
  Then, we use the variant of Dijkstra's algorithm mentioned in \Cref{prop:lapprox} to find the index~$i\in I$ and a shortest $s_i$-$t_i$-path~$P$ in~$G[W]$ with a minimum number of arcs.
  Finally, we set~$f[X,1]$ to be equal to the number of arcs in~$P$.

  For $r\geq 2$, we compute $f[X,r]$ for each~$X\subseteq [p+\ell]$ using the recurrence relation 
  \begin{equation}\label{eq:rec}
  f[X,r]=\min_{Y\subset X}\{f[X\setminus Y,r-1]+f[Y,1]\}.
  \end{equation}

  The correctness of computing the values of $f[X,1]$ follows from the description and the correctness of recurrence~(\ref{eq:rec}) follows from the condition that distinct paths should not have vertices of the same color (including their ends).

  We compute the values~$f[X,r]$ in order of increasing~$r \in [p]$.
  Since computing~$f[Y,1]$ for a given set~$Y$ of colors can be done in polynomial time, we can compute all values in overall~$3^{p+\ell}\poly(n)$~time.
  Once all values~$f[X,r]$ are computed, we observe that a colorful solution exists if and only if~$f[S,p]\leq \ell$.

  If there is a colorful solution, then we conclude that $(G,w,(s_1,t_1),\dots,(s_k,t_k),p)$ is a yes-instance of \mdsp.  
  Otherwise, we discard the considered coloring and try another random coloring and iterate.
  If we fail to find a solution after executing $N=\lceil e^{p+\ell}\rceil$ iterations, we obtain that the probability that~$(G,w,(s_1,t_1),\dots,(s_k,t_k),p)$ is a yes-instance is at most~$(1-\frac{1}{e^{p+\ell}})^{e^{p+\ell}}\leq e^{-1}$. Thus, we return that~$(G,w,(s_1,t_1),\dots,(s_k,t_k),p)$ is a no-instance with the error probability upper bounded by $e^{-1}<1$. 
  Since the running time in each iteration is~$3^{p+\ell}\poly(n)$ and~$p\leq \ell$, the total running time is in~$2^{O(\ell)}\poly(n)$.
  Note that we do the color coding and dynamic programming for each value between~$p$ and the actual value~$\ell$.
  However, this only adds an additional factor of~$\ell \leq n$ which disappears in the~$\poly(n)$.
  
  The above algorithm can be derandomized using the results of Naor, Schulman, and Srinivasan~\cite{NSS95} by replacing random colorings by prefect hash families. We refer to the textbook by Cygan et al.~\cite{CyganFKLMPPS15} for details on this common technique.
\end{proof}
}

The fixed-parameter tractability result of \cref{prop:fptl} immediately raises two questions: Can the running time be significantly improved and does \mdsp{} parameterized by~$\ell$ admit a polynomial kernel?
We next show that the answer to both questions is no.
We start with an ETH-based lower bound.
Recall that \dsp{} is the special case of \mdsp{} where the input graph is unweighted and all paths in the solution need to be shortest paths between the respective ends.
Before we give the proof, we first show a intermediate result that will allow us to extend our hardness result to planar graphs.
The intermediate result is achieved by combining a number of known reductions.
Since we could not find this combination of results in the existing literature, we give a proof sketch for the sake of completeness.
We start by defining a few relevant versions of \textsc{Satisfiability}.

For a formula~$\phi$ with variables~$x_1,x_2,\ldots,x_n$ and clauses~$C_1,C_2,\ldots,C_m$, we say that the \emph{variable-clause graph associated with~$\phi$} has one vertex~$v_i$ for each variable~$x_i$, one vertex~$u_j$ for each clause~$C_j$, and an edge between~$v_i$ and~$u_j$ if and only if~$x_i$ or~$\neg x_i$ appear in~$C_j$.
\textsc{Planar 3-Sat} is the problem of deciding whether a formula~$\phi$ of \textsc{3-Sat} with a planar associated graph is satisfiable.
An instance of \textsc{Sat} is \emph{monotone}, if each clause is monotone, that is, it either contains only positive literals or only negative literals.
Finally, \textsc{Rectilinear 3-Sat} is a special case of \textsc{Planar 3-Sat} in which we additionally require that adding a variable cycle still results in a planar associated graph.
A \emph{variable cycle} is an additional set of edges that induces a cycle through all vertices~$v_i$.
More specifically in \textsc{Rectilinear 3-Sat}, one is given a planar embedding of the graph~$G_\phi$ associated with~$\phi$ where all vertices are drawn as rectangles, all edges are vertical segments, and all rectangles corresponding to variables are drawn on a horizontal line and no rectangles corresponding to a clause intersects this line.
In \textsc{Rectilinear Monotone 3-Sat}, all rectangles corresponding to positive clauses are drawn above the variable rectangles and all rectangles corresponding to negative clauses are drawn below the variable rectangles.
An example of \textsc{Rectilinear Monotone 3-Sat} is given in \cref{fig:examplepm3}.  
\begin{figure}[t]
    \centering
    \begin{tikzpicture}
        \node[rectangle,draw,minimum height=.5cm, minimum width=1cm] at(0,0) (x1) {$v_1$};
        \node[rectangle,draw,minimum height=.5cm, minimum width=1cm] at(2,0) (x2) {$v_2$};
        \node[rectangle,draw,minimum height=.5cm, minimum width=1cm] at(4,0) (x3) {$v_3$};
        \node[rectangle,draw,minimum height=.5cm, minimum width=1cm] at(6,0) (x4) {$v_4$};
        \node[rectangle,draw,minimum height=.5cm, minimum width=1cm] at(8,0) (x5) {$v_5$};
        \node[rectangle,draw,minimum height=.5cm, minimum width=4cm] at(2,1) (C1) {$u_1$};
        \node[rectangle,draw,minimum height=.5cm, minimum width=5.5cm] at(3,-1) (C2) {$u_2$};
        \node[rectangle,draw,minimum height=.5cm, minimum width=8cm] at(4,-2) (C3) {$u_3$};

        \draw (.25,.25) to (.25,.75);
        \draw (2,.25) to (2,.75);
        \draw (3.725,.25) to (3.725,.75);
        
        \draw (.375,-.25) to (.375,-.75);
        \draw (4,-.25) to (4,-.75);
        \draw (5.75,-.25) to (5.75,-.75);
        
        \draw (.125,-.25) to (.125,-1.75);
        \draw (6,-.25) to (6,-1.75);
        \draw (8,-.25) to (8,-1.75);
    \end{tikzpicture}
    \captionsetup{singlelinecheck=off}
    \caption{The graph~$G_\phi$ associated with the formula 
    {\[\phi = (x_1 \lor x_2 \lor x_3) \land (\neg x_1 \lor \neg x_3 \lor \neg x_4) \land (\neg x_1 \lor \neg x_4 \lor \neg x_5).\]}}
    \label{fig:examplepm3}
\end{figure}
We are now in a position to state the intermediate result.

\begin{restatable}[folklore]{proposition}{SAT}
    \label{prop:SAT}
    Assuming the ETH, \textsc{Rectilinear Monotone 3-Sat} cannot be solved in~$2^{o(\sqrt{n+m})}$ time even if each variable appears in at most~$8$ clauses.
\end{restatable}

\begin{proof}[Proof sketch]
    We start from \textsc{3-Sat} cannot be solved in~$2^{o(n+m)}$ time unless the ETH fails~\cite{IP01,IPZ01}.
    Using a reduction by Tovey~\cite{Tov84}, we reduce \textsc{3-Sat} in polynomial time to an equivalent instance of~\textsc{3-Sat} in which each variable appears in at most~$4$ clauses.
    The number of variables and clauses grows only by a linear factor in this reduction and therefore this version of \textsc{3-Sat} can also not be solved in~$2^{o(n+m)}$ time unless the ETH fails.
    Next, we use a polynomial-time reduction due to Lichtenstein~\cite{Lic82} to transform the previous instance into an instance of planar \textsc{3-Sat} where each variable appears in at most~$7$ clauses.
    Moreover, the reduction also outputs an ordering of the variables~$x_1,x_2,\ldots,x_n$ such that adding the edges in~${\{\{x_i,x_i+1\} \mid i < n\} \cup \{x_n,x_1\}}$ to the variable-clause graph associated with the input formula still results in a planar graph.
    The number of clauses increases by a quadratic factor and hence we get an ETH-based lower bound of~$2^{o(\sqrt{n+m})}$.
    It was observed by Knuth and Raghunathan~\cite{KR92} that the instance produced by the reduction by Lichtenstein is an instance of \textsc{Rectilinear 3-Sat}.
    Finally, using a reduction due to de Berg and Khosravi~\cite{dBK12}, we turn the instance of \textsc{Rectilinear 3-Sat} where each variable appears in at most~$7$ clauses to an instance of \textsc{Rectilinear Monotone 3-Sat} where each variable appears in at most~$8$ clauses.
    In this reduction, the number of variables and clauses again increases by a linear factor and hence we get a~$2^{o(\sqrt{n+m})}$-time ETH-based lower bound.
\end{proof}

We next show that \dsp{} cannot be solved in~$2^{o(n+m)}$ time and, on planar graphs, it cannot be solved in~$2^{o(n+m)} = 2^{o(n)}$ time.
Note that this also excludes~$2^{o(\ell)}\poly(n)$-time and~$2^{o(\sqrt{\ell})}\poly(n)$-time algorithms for \dsp{} (on planar graphs) since~$\ell < n$ and~$\ell \leq m$.
Since \dsp{} is a special case of \mdsp, we also get the same lower bounds for the latter.

\begin{restatable}{proposition}{mainETH}
\label{thm:main-ETH}
    \dsp{} cannot be solved in~$2^{o(\sqrt{n+m})}$ time unless the ETH fails.
    Under the same assumption, it cannot be solved in~$2^{o(\sqrt{n})}$ time on planar graphs.
\end{restatable}

{
\begin{proof}
    For general graphs, we reduce from \textsc{3-Sat}, which cannot be solved in~$2^{o(n+m)}$ time unless the ETH fails~\cite{IP01,IPZ01}.
    For planar graphs, we reduce from \textsc{Rectilinear Monotone 3-Sat} with each variable appearing in at most~$8$ clauses.
    This problem cannot be solved in~$2^{o(\sqrt{n+m})}$ time unless the ETH fails by \cref{prop:SAT}.
    Otherwise, the reductions are the same except for the fact that we ignore the embedding for general graphs.
    To avoid repetition, we will present the reduction only for planar graphs.
    We will have a terminal pair~$(s_i,t_i)$ for each variable~$x_i$ and at most two terminal pairs~$(s'_j,t'_j)$ and~$(s^*_j,t^*_j)$ for each clause~$C_j$ in the input formula~~$\phi$.
    First, we replace each of the rectangular vertices representing a variable~$x_i$ as follows.
    For each edge incident to the rectangle, we place one vertex on the boundary (at the place of the intersection).
    We then place~$s_i$ on the left boundary,~$t_i$ on the right boundary, and additional dummy vertices on either the upper or lower boundary until both boundaries have the same number of vertices (at most~$8$).
    We connect all vertices on the boundary by a cycle that follows the boundary.
    We say that the path from~$s_i$ to~$t_i$ via vertices on the upper boundary is the upper paths and the path consisting of vertices on the lower boundary as the lower path.
    An example for~$v_1$ in \cref{fig:examplepm3} is given in \cref{fig:vg}.
    \begin{figure}
        \centering
        \begin{tikzpicture}
            \node[circle,draw,label=left:$s_1$] at(0,0) (s) {};
            \node[circle,draw] at(2,1) (u1) {};
            \node[circle,draw] at(3,1) (u2) {};
            \node[circle,draw] at(1,-1) (l1) {};
            \node[circle,draw] at(4,-1) (l2) {};
            \node[circle,draw,label=right:$t_1$] at(5,0) (t) {};

            \draw (s) to (0,1);
            \draw (s) to (0,-1);
            \draw (0,1) to (u1);
            \draw (0,-1) to (l1);
            \draw (u1) to (u2);
            \draw (l1) to (l2);
            \draw (u2) to (5,1);
            \draw (l2) to (5,-1);
            \draw (5,1) to (t);
            \draw (5,-1) to (t);
        \end{tikzpicture}
        \caption{The gadget constructed for vertex~$v_1$ in \cref{fig:examplepm3}. One of the two highest vertices is a dummy vertex.}
        \label{fig:vg}
    \end{figure}

    Next, we replace the vertices representing clauses.
    We distinguish between clauses with one, two, and three literals in them (and we assume that each variable appears at most once in each clause).
    For a clause~$C_j$ with one or two literals, the gadget simply consists of two terminal vertices~$s'_j$ and~$t'_j$ which are connected to their respective neighbors in the clause gadgets (the vertices constructed on the boundary of the variable rectangle because of an edge to clause~$C_j$) by paths of length~$5$ (which can either be implemented using edge weights or by subdividing an edge 4 times).
    If a clause~$C_j$ contains three literals, then we construct the gadget displayed in \cref{fig:cg}.
    \begin{figure}
        \centering
        \begin{tikzpicture}
            \node[circle,draw,label=below:$a$] at(0,0) (a) {};
            \node[circle,draw,label=below:$b$] at(3,0) (b) {};
            \node[circle,draw,label=below:$c$] at(6,0) (c) {};
            \node[circle,draw] at(3,4) (v) {};
            \node[circle,draw] at(1,2) (w) {};
            \node[circle,draw,label=below:$v$] at(4.5,2) (x) {};
            \node[circle,draw] at(3.5,2.5) (y) {};
            \node[circle,draw] at(5.5,2.5) (z) {};
            \node[circle,draw,label=$s'_j$] at(2,2) (s1) {};
            \node[circle,draw,label=left:$s^*_j$] at(4.5,3.5) (s2) {};
            \node[circle,draw,label=right:$t'_j$] at(6,4) (t1) {};
            \node[circle,draw,label=$t^*_j$] at(4.5,2.5) (t2) {};

            \draw (a)+(-.12,.12) to (-.12,4);
            \draw (a)+(.12,.12) to (0.12,2);
            \draw (.12,2) to (w);
            \draw (-.12,4) to (v);
            \draw (v) to (t1);
            \draw[very thick] (t1)+(.12,-.12) to (6.12,2);
            \draw[very thick] (c)+(.12,.12) to (6.12,2);
            \draw (w) to (s1);
            \draw (s1) to (2.9,2);
            \draw (b)+(-.12,.12) to (2.9,2);
            \draw (b)+(.12,.12) to (3.12,2);
            \draw (3.12,2) to (x);
            \draw[very thick] (c)+(-.12,.12) to (5.9,2);
            \draw[very thick] (5.9,2) to (x);
            \draw[very thick] (y) to (x);
            \draw[very thick] (y) to (s1);
            \draw (y) to (s2);
            \draw (y) to (t2);
            \draw (z) to (x);
            \draw (z) to (t1);
            \draw[very thick] (z) to (s2);
            \draw[very thick] (z) to (t2);

            \node[rectangle,dotted,draw, minimum width=8cm, minimum height=3cm] at (3,3) {};
        \end{tikzpicture}
        \caption{The gadget constructed for a clause with three literals. The vertices~$a,b,$ and~$c$ are the three corresponding neighbors in variable gadgets of vertices in the clause gadget. Each edge represents a path of length~$5$ and the thick edges show the shortest paths between the terminals that are used if vertex~$c$ is free (not used in a solution path in the corresponding variable gadget).}
        \label{fig:cg}
    \end{figure}    
    
    Note that the reduction can be computed in polynomial time and contains~$O(m)$ vertices as each gadget has constant size and the number of variables in any \textsc{3-Sat}-instance is at most~$3m$.
    Moreover, if we start with an instance of \textsc{Rectilinear Monotone 3-Sat}, then the constructed instance is planar as each gadget is planar and the edges between gadgets can be drawn arbitrarily close to the respective edges in~$G_\phi$.
    Hence, it only remains to prove that the constructed instance is a yes-instance if and only if~$\phi$ is satisfiable.
    To this end, first assume that~$\phi$ is satisfiable.
    Then, there exists some satisfying assignment~$\beta$ for~$\phi$.
    For each variable~$x_i$, if~$\beta$ assigns true to~$x_i$, then we connect~$s_i$ to~$t_i$ via the lower path.
    Otherwise, we connect~$s_i$ to~$t_i$ via the upper path.
    Note that both of these paths have the same length and these paths are indeed shortest paths between~$s_i$ and~$t_i$ as they have length at most~$9$ and leaving the gadget for~$v_i$ results in a path of length at least~$10$ as each path between a vertex in a variable gadget and a vertex in a clause gadget has length~$5$.
    Consider now the gadget for any clause~$C_j$.
    Since~$\beta$ is a satisfying assignment for~$\phi$, there is a variable satisfying~$C_j$.
    Hence by construction, the respective neighbor of a vertex in the gadget is not used to connect the corresponding terminals in the variable gadget.
    If~$C_j$ contains one or two literals, then there is a shortest path for the two terminal vertices in the clause gadget that can be used to connect them.
    If~$C_j$ contains three literals, then there are also vertex-disjoint shortest paths to connect both terminal pairs as shown next.
    Consider the gadget in \cref{fig:cg} and first assume that the vertex~$a$ is free (not used to connect the two terminals in the corresponding variable gadget).
    Then, we can connect~$s'_j$ to~$t'_j$ using the unique shortest path going through~$a$ and connect~$s^*_j$ to~$t^*_j$ via either of the two shortest paths.
    Next, assume that the vertex~$b$ is free.
    Then, we connect~$s'_j$ to~$t'_j$ by using the two paths incident to vertex~$b$ and then finish the path by the unique shortest path between~$v$ and~$t'_j$ that stays within the clause gadget.
    We connect~$s^*_j$ to~$t^*_j$ by the left shortest path.
    Lastly, assume that~$c$ is free.
    Then, we connect~$s'_j$ to~$t'_j$ by using the unique shortest path from~$s'_j$ to~$v$ that stays within the clause gadget and then the two paths incident to vertex~$c$.
    We connect~$s^*_j$ to~$t^*_j$ by the right shortest path.
    These two paths are also highlighted in \cref{fig:cg}.
    Note that all described paths are indeed shortest paths as the three vertices~$a$,~$b$ and~$c$ correspond to different variables and hence they have pairwise distance at least~$10$.

    For the reverse direction, assume that there is a solution to the constructed instance of \dsp.
    We first show that each clause uses at least one vertex from a vertex gadget.
    This statement is trivially true for clauses with one or two literals.
    Now consider a terminal pair~$(s'_j,t'_j)$ in a clause gadget for a clause with three literals.
    Note that all paths between~$s'_j$ and~$t'_j$ that do not leave the clause gadget contain both neighbors of~$s^*_j$.
    Thus, if the shortest~\mbox{$s'_j$-$t'_j$-path} does not leave the clause gadget, then there is no disjoint~$s^*_j$-$t^*_j$-path.

    As shown above, there are only two shortest paths between the two terminals in a variable gadget.
    Thus, for each variable gadget, the solution picks for each variable either the upper or the lower path in the variable gadget to connect the terminals in such a way that for each clause gadget at least one of the three neighbors in variable gadgets is free.
    Let~$\beta$ be the assignment that assigns true to~$x_i$ if and only if the solution connects~$s_i$ to~$t_i$ via the lower path.
    By construction this assignment satisfies each clause and therefore confirms that~$\phi$ is satisfiable.
    This concludes the proof.
\end{proof}
}

We next exclude a polynomial kernel for \dsp{} parameterized by~$\ell$.
To show that a parameterized problem~$P$ does (presumably) not admit a polynomial kernel, one can use the framework of \emph{cross-compositions}.
Given an NP-hard problem~$L$, a polynomial equivalence relation~$R$ on the instances of~$L$ is an equivalence relation such that (i) one can decide for any two instances in polynomial time whether they belong to the same equivalence class, and (ii) for any finite set~$S$ of instances,~$R$ partitions the set into at most~$\max_{I \in S} \poly(|I|)$ equivalence classes.
Given an NP-hard problem~$L$, a parameterized problem~$P$, and a polynomial equivalence relation~$R$ on the instances of~$L$, an OR-cross-composition of~$L$ into~$P$ (with respect to~$R$) is an algorithm that
takes~$q$~instances~$I_1,I_2,\ldots,I_q$ of~$L$ belonging to the same equivalence class of~$R$ and constructs in~$\poly(\sum_{i=1}^{q} |I_i|)$~time an instance~$(I,\rho)$ of~$P$ such that (i) $\rho$ is polynomially upper-bounded by~$\max_{i \in [q]} |I_i| + \log q$, and (ii) $(I,\rho)$ is a yes-instance of~$P$ if and only if at least one of the instances~$I_i$ is a yes-instance of~$L$.
If a parameterized problem admits an OR-cross-composition, then it does not admit a polynomial kernel unless NP~$\subseteq$ coNP/poly~\cite{BJK14}.

In order to exclude a polynomial kernel, we first show that a special case of \mdsp{} remains NP-hard.
We call this special case \ldsp{} and it is the special case of \dsp{} where all edges have weight one and the input graph is layered, that is, there is a partition of the vertices into (disjoint) sets~$V_1,V_2,\ldots,V_{\lambda}$ such that all edges~$\{u,v\}$ are between two consecutive layers, that is~$u \in V_i$ and~${v \in V_{i+1}}$ or~$u \in V_{i+1}$ and~$v \in V_{i}$ for some~$i \in [\lambda-1]$.
Moreover, each terminal pair~$(s_i,t_i)$ satisfies that~$s_i \in V_1$, $t_i \in V_\lambda$, and each shortest path between the two terminals is \emph{monotone}, that is, it contains exactly one vertex of each layer.
\ldsp{} is formally defined as follows.

\defproblem{\ldsp}
{A~$\lambda$-layered graph~$G=(V,E)$ with a $\lambda$-partition $\{V_1,V_2,\ldots,V_\lambda\}$ of the vertex set,
terminal pairs~$(s_1,t_1), (s_2,t_2), \ldots, (s_k,t_k)$ with~$s_i \in V_1$, $t_i \in V_{\lambda}$, and~$\dist(s_i,t_i) = \lambda-1$ for all~$i \in [\lambda]$.}
{Is there a collection~$\mathcal{C}=\{P_i\}_{i\in [k]}$ of pairwise vertex-disjoint paths such that~$P_i$ is an~$s_i$-$t_i$-path of length~$\lambda-1$ for all~$i \in [k]$?}
It is quite straight-forward to prove that \ldsp{} is NP-complete as shown next.

\begin{restatable}{proposition}{lnp} 
    \label{prop:lnp}
    \ldsp{} is NP-complete.
\end{restatable}

{

\begin{proof}
    We focus on the NP-hardness as~\ldsp{} is a special case of \dsp{} and therefore clearly in NP.
    We reduce from \textsc{3-Sat}.
    The main part of the reduction is a selection gadget.
    The gadget consists of a set~$U$ of~$n+1$ vertices~$u_0,u_1,\ldots,u_n$ and between each pair of consecutive vertices~$u_{i-1},u_i$, there are two paths with~$m$ internal vertices each.
    Let the set of vertices be~$V_i = \{v_1^i,v_2^i,\ldots,v_m^i\}$ and~$W_i = \{w_1^i,w_2^i,\ldots,w_m^i\}$.
    The set of edges in the selection gadget is
    \begin{align*}
        E = \{\{u_{i-1},v_1^i\},\{u_{i-1},w_1^i\}, \{v_{m}^i,u_i\},\{w_m^{i},u_i\} &\mid i \in [n]\} \\ \cup\,\{\{v_j^i,v_{j+1}^i\}, \{w_j^i,w_{j+1}^i\}&\mid i \in [n] \land j \in [m-1]\}.
    \end{align*}
    The constructed instance will have~$m+1$ terminal pairs and is depicted in \cref{fig:layered}.
    \begin{figure}[t]
		\centering
		\begin{tikzpicture}
            \node[circle,draw,label=$s_3$] (s0) at (0,4.5) {};
            \node[circle,draw] (v1) at (-.5,4) {} edge(s0);
            \node[circle,draw] (v2) at (-.5,3.5) {} edge(v1);
            \node[circle,draw] (w1) at (.5,4) {} edge(s0);
            \node[circle,draw] (w2) at (.5,3.5) {} edge(w1);
            \node[circle,draw,label=left:$u_1$] (u1) at (0,3) {} edge(v2) edge(w2);
            \node[circle,draw] (v3) at (-.5,2.5) {} edge(u1);
            \node[circle,draw] (v4) at (-.5,2) {} edge(v3);
            \node[circle,draw] (w3) at (.5,2.5) {} edge(u1);
            \node[circle,draw] (w4) at (.5,2) {} edge(w3);
            \node[circle,draw,label=left:$u_2$] (u2) at (0,1.5) {} edge(v4) edge(w4);
            \node[circle,draw] (v5) at (-.5,1) {} edge(u2);
            \node[circle,draw] (v6) at (-.5,.5) {} edge(v5);
            \node[circle,draw] (w5) at (.5,1) {} edge(u2);
            \node[circle,draw] (w6) at (.5,.5) {} edge(w5);
            \node[circle,draw,label=below:$t_3$] at (0,0) {} edge(v6) edge(w6);

            \node[circle,draw,label=$s_1$] (s1) at (-2.5,4.5) {} edge(w1);

            \node[circle,draw] (a2) at (-3,3.5) {} edge(w1);
            \node[circle,draw] (a3) at (-3,3) {} edge(a2);
            \node[circle,draw] (a4) at (-3,2.5) {} edge(a3);
            \node[circle,draw] (a5) at (-3,2) {} edge(a4);
            \node[circle,draw] (a6) at (-3,1.5) {} edge(a5);
            \node[circle,draw] (a7) at (-3,1) {} edge(a6);
            \node[circle,draw] (a8) at (-3,.5) {} edge(a7);
            
            \node[circle,draw] (b1) at (-2.5,4) {} edge(s1);
            \node[circle,draw] (b2) at (-2.5,3.5) {} edge(b1);
            \node[circle,draw] (b3) at (-2.5,3) {} edge(b2) edge(w3);
            \node[circle,draw] (b5) at (-2.5,2) {} edge(w3);
            \node[circle,draw] (b6) at (-2.5,1.5) {} edge(b5);
            \node[circle,draw] (b7) at (-2.5,1) {} edge(b6);
            \node[circle,draw] (b8) at (-2.5,.5) {} edge(b7);
            
            \node[circle,draw] (c1) at (-2,4) {} edge(s1);
            \node[circle,draw] (c2) at (-2,3.5) {} edge(c1);
            \node[circle,draw] (c3) at (-2,3) {} edge(c2);
            \node[circle,draw] (c4) at (-2,2.5) {} edge(c3);
            \node[circle,draw] (c5) at (-2,2) {} edge(c4);
            \node[circle,draw] (c6) at (-2,1.5) {} edge(c5) edge(v5);
            \node[circle,draw] (c8) at (-2,.5) {} edge(v5);
            
            \node[circle,draw,label=below:$t_1$] at (-2.5,0) {} edge(a8) edge(b8) edge(c8);

            \node[circle,draw,label=$s_2$] (s2) at (2.5,4.5) {};

            \node[circle,draw] (d1) at (2,4) {} edge(s2) edge(w2);
            \node[circle,draw] (d3) at (2,3) {} edge(w2);
            \node[circle,draw] (d4) at (2,2.5) {} edge(d3);
            \node[circle,draw] (d5) at (2,2) {} edge(d4);
            \node[circle,draw] (d6) at (2,1.5) {} edge(d5);
            \node[circle,draw] (d7) at (2,1) {} edge(d6);
            \node[circle,draw] (d8) at (2,.5) {} edge(d7);
            
            \node[circle,draw] (e1) at (2.5,4) {} edge(s2);
            \node[circle,draw] (e2) at (2.5,3.5) {} edge(e1);
            \node[circle,draw] (e3) at (2.5,3) {} edge(e2);
            \node[circle,draw] (e4) at (2.5,2.5) {} edge(e3) edge(v4);
            \node[circle,draw] (e6) at (2.5,1.5) {} edge(v4);
            \node[circle,draw] (e7) at (2.5,1) {} edge(e6);
            \node[circle,draw] (e8) at (2.5,.5) {} edge(e7);
            
            \node[circle,draw] (f1) at (3,4) {} edge(s2);
            \node[circle,draw] (f2) at (3,3.5) {} edge(f1);
            \node[circle,draw] (f3) at (3,3) {} edge(f2);
            \node[circle,draw] (f4) at (3,2.5) {} edge(f3);
            \node[circle,draw] (f5) at (3,2) {} edge(f4);
            \node[circle,draw] (f6) at (3,1.5) {} edge(f5);
            \node[circle,draw] (f7) at (3,1) {} edge(f6) edge(w6);
            
            \node[circle,draw,label=below:$t_2$] at (2.5,0) {} edge(d8) edge(e8) edge(w6);

        \end{tikzpicture}    \captionsetup{singlelinecheck=off}
        \caption{An example of the construction in the proof of \cref{prop:lnp} for the input formula \[\phi = (x_1 \lor x_2 \lor \overline{x_3}) \land (x_1 \lor \overline{x_2} \lor x_3).\]}
        \label{fig:layered}
    \end{figure}
    We set~$s_{m+1} = u_0$ and~$t_{m+1} = u_n$ and we will ensure that any shortest~$s_{m+1}$-$t_{m+1}$-path contains all vertices in~$U$ and for each~$i \in [n]$ either all vertices in~$V_i$ or all vertices in~$W_i$.
    These choices will correspond to setting the~$i$\textsuperscript{th}~variable to either true or false.
    Additionally, we have a terminal pair~$(s_j,t_j)$ for each clause~$C_j$.
    There are (up to) three disjoint paths between~$s_j$ and~$t_j$, each of which is of length~$n \cdot (m+1)$.
    These paths correspond to which literal in the clause satisfies it.
    For each of these paths, let~$x_i$ be the variable corresponding to the path.
    If~$x_i$ appears positively in~$C_j$, then we identify the~$(i-1)(m+1)+j+1$\textsuperscript{st} vertex in the path with~$w_j^i$ and if~$x_i$ appears negatively, then we identify the vertex with~$v_j^i$.
    Note that the constructed instance is~$(n\cdot (m+1))$-layered and that once any monotone path starting in~$s_{m+1}$ leaves the selection gadget, it cannot end in~$t_{m+1}$ as any vertex outside the selection gadget has degree at most two and at the end of these paths are only terminals~$t_1,t_2,\ldots,t_{m}$.

    Since the construction runs in polynomial time, we focus on the proof of correctness.
    If the input formula is satisfiable, then we connect all terminal pairs as follows.
    Let~$\beta$ be a satisfying assignment.
    The terminal pair~$(s_{m+1},t_{m+1})$ is connected by a path containing all vertices in~$U$ and for each~$i \in [n]$, if~$\beta$ assigns the~$i$\textsuperscript{th} variable to true, then the path contains all vertices in~$V_i$ and otherwise all vertices in~$W_i$.
    For each clause~$C_j$, let~$x_{i_j}$ be variable in~$C_j$ which~$\beta$ uses to satisfy~$C_j$ (if multiple such variables exist, we choose an arbitrary one).
    By construction, there is a path associated with~$x_{i_j}$ that connects~$s_j$ and~$t_j$ and only uses one vertex in~$W_i$ if~$x_{i_j}$ appears positively in~$C_j$ and a vertex in~$V_i$, otherwise.
    Since each vertex in~$V_i$ and~$W_i$ is only associated with at most one such path, we can connect all terminal pairs.
    For the other direction assume that all~$m+1$ terminal pairs can be connected by disjoint shortest paths.
    As argued above, the~$s_{m+1}$-$t_{m+1}$-path stays in the selection gadget.
    We define a truth assignment by assigning the~$i$\textsuperscript{th} variable to true if and only if the~$s_{m+1}$-$t_{m+1}$-path contains the vertices in~$V_i$.
    For each clause~$C_j$, we look at the neighbor of~$s_j$ in the solution.
    This vertex belongs to a path of degree-two vertices that at some point joins the selection gadget.
    By construction, the vertex where this happens is not used by the~$s_{m+1}$-$t_{m+1}$-path, which guarantees that~$C_j$ is satisfied by the corresponding variable.
    Since all clauses are satisfied by the same assignment, the formula is satisfiable and this concludes the proof.
\end{proof}
}

With the NP-hardness of \ldsp{} at hand, we can now show that it does not admit a polynomial kernel when parameterized by~$\ell$ by providing an OR-cross-composition from its unparameterized version to the version parameterized by~$\ell$.

\begin{theorem}
    \label{thm:nopkl}
    \ldsp{} parameterized by $\ell = k\cdot (\lambda-1)$ does not admit a polynomial kernel unless \npconp.
\end{theorem}

\begin{proof}
    We present an OR-cross-composition from \ldsp{} into \ldsp{} parameterized by~$\ell$.
    To this end, assume we are given~$t$ instances of \ldsp{} all of which have the same number~$\lambda$ of layers and the same number~$k$ of terminal pairs.
    Moreover, we assume that~$t$ is some power of two.
    Note that we can pad the instance with at most~$t$ trivial no-instances to reach an equivalent instance in which the number of instances is a power of two and the size of all instances combined has at most doubled.

    The main ingredient for our proof is a construction to merge two instances into one.
    The construction is depicted in \cref{fig:nopolykernel}.
    We first prove that the constructed instance is a yes-instance if and only if at least one of the original instances is a yes-instance.
    Afterwards, we will show how to use this construction to get an OR-cross-composition for all~$t$ instances.
    \begin{figure}[t]
        \centering
        \begin{tikzpicture}[yscale=.5]
            \node[circle,draw,label=$s_1$] at(3.5,9) (v11) {};
            \node[circle,draw,label=$s_2$] at(4.5,9) (v12) {};
            \node[circle,draw,label=$s_3$] at(5.5,9) (v13) {};
            \node at(7,9) {$\dots$};
            \node[circle,draw,label=$s_k$] at(8.5,9) (v15) {};

            \node[circle,draw] at(3,8) (v21) {} edge(v11);
            \node[circle,draw] at(4,8) (v22) {} edge(v12) edge(v11);
            \node[circle,draw] at(5,8) (v23) {} edge(v13) edge(v12);
            \node[circle,draw] at(7,8) (v25) {} edge[dotted](v15) edge[dotted](v13);
            \node[circle,draw] at(9,8) (v26) {} edge(v15);
            
            \node[circle,draw] at(2.5,7) (v31) {} edge(v21);
            \node[circle,draw] at(3.5,7) (v32) {} edge(v22);
            \node[circle,draw] at(4.5,7) (v33) {} edge(v23) edge(v22);
            \node[circle,draw] at(6.5,7) (v35) {} edge[dotted](v23) edge(v25);
            \node[circle,draw] at(7.5,7) (v36) {} edge(v25);
            \node[circle,draw] at(9.5,7) (v37) {} edge(v26);
            
            \node[circle,draw] at(2,6) (v41) {} edge[dotted](v31);
            \node[circle,draw] at(3,6) (v42) {} edge[dotted](v32);
            \node[circle,draw] at(4,6) (v43) {} edge[dotted](v33);
            \node[circle,draw] at(6,6) (v45) {} edge(v35) edge[dotted] (v33);
            \node[circle,draw] at(7,6) (v46) {} edge(v35);
            \node[circle,draw] at(8,6) (v47) {} edge(v36);
            \node[circle,draw] at(10,6) (v48) {} edge[dotted](v37);
        
            \node[circle,draw,label=left:$s_1^i$] at(1,5) (a) {} edge(v41);
            \node[circle,draw,label=left:$s_2^i$] at(2,5) (b) {} edge(v42);
            \node[circle,draw,label=left:$s_3^i$] at(3,5) (c) {} edge(v43);
            \node[] at(3.8,5) {$\dots$};
            \node[circle,draw,label=left:$s_k^i$] at(5,5) (d) {} edge(v45);
            \node at(0,3) {$G_i$};
            \node[circle,draw,label=left:$t_1^i$] at(5,1) (e) {};
            \node[circle,draw,label=left:$t_2^i$] at(4,1) (f) {};
            \node[circle,draw,label=left:$t_3^i$] at(3,1) (g) {};
            \node[] at(1.8,1) {$\dots$};
            \node[circle,draw,label=left:$t_k^i$] at(1,1) (h) {};

            \node[circle,draw,label=left:$s_1^j$] at(7,5) (A) {} edge(v45);
            \node[circle,draw,label=left:$s_2^j$] at(8,5) (B) {} edge(v46);
            \node[circle,draw,label=left:$s_3^j$] at(9,5) (C) {} edge(v47);
            \node[] at(9.8,5) {$\dots$};
            \node[circle,draw,label=left:$s_k^j$] at(11,5) (D) {} edge(v48);
            \node at(11.6,3) {$G_j$};
            \node[circle,draw,label=left:$t_1^j$] at(11,1) (E) {};
            \node[circle,draw,label=left:$t_2^j$] at(10,1) (F) {};
            \node[circle,draw,label=left:$t_3^j$] at(9,1) (G) {};
            \node[] at(7.8,1) {$\dots$};
            \node[circle,draw,label=left:$t_k^j$] at(7,1) (H) {};
            
            \node[rectangle,rounded corners,dotted,draw, minimum width=5cm, minimum height=2.5cm] at (2.8,3) {};
            \node[rectangle,rounded corners,dotted,draw, minimum width=5cm, minimum height=2.5cm] at (8.8,3) {};
            
            \node[circle,draw] at(2,0) (v51) {} edge(h);
            \node[circle,draw] at(4,0) (v52) {} edge(g);
            \node[circle,draw] at(5,0) (v53) {} edge(f);
            \node[circle,draw] at(6,0) (v55) {} edge(e) edge(H);
            \node[circle,draw] at(8,0) (v56) {} edge(G);
            \node[circle,draw] at(9,0) (v57) {} edge(F);
            \node[circle,draw] at(10,0) (v58) {} edge(E);
            
            \node[circle,draw] at(2.5,-1) (v61) {} edge[dotted](v51);
            \node[circle,draw] at(4.5,-1) (v62) {} edge(v52);
            \node[circle,draw] at(5.5,-1) (v63) {} edge(v53) edge(v55);
            \node[circle,draw] at(7.5,-1) (v65) {} edge[dotted](v56) edge[dotted](v55);
            \node[circle,draw] at(8.5,-1) (v66) {} edge[dotted](v57);
            \node[circle,draw] at(9.5,-1) (v67) {} edge[dotted](v58);
            
            \node[circle,draw] at(3,-2) (v71) {} edge(v61);
            \node[circle,draw] at(5,-2) (v72) {} edge(v62) edge(v63);
            \node[circle,draw] at(7,-2) (v73) {} edge[dotted](v63) edge(v65);
            \node[circle,draw] at(8,-2) (v75) {} edge(v65) edge(v66);
            \node[circle,draw] at(9,-2) (v76) {} edge(v67);
            
            \node[circle,draw,label=$t_1$] at(8.5,-3) (v81) {} edge(v75) edge(v76);
            \node[circle,draw,label=$t_2$] at(7.5,-3) (v82) {} edge(v75) edge(v73);
            \node[circle,draw,label=$t_3$] at(6.5,-3) (v83) {} edge(v73) edge[dotted](v72);
            \node at(5,-3) {$\dots$};
            \node[circle,draw,label=$t_k$] at(3.5,-3) (v85) {} edge(v71) edge[dotted](v72);
        \end{tikzpicture}
        \caption{The construction to merge two instances of \ldsp{} into one equivalent instance. The dotted edges can be read as regular edges for~$k=4$ and indicate where additional vertices and edges have to be added for more terminal pairs. Note that the height of a vertex in the drawing does not indicate its layer as dotted edges distort the picture.}
        \label{fig:nopolykernel}
    \end{figure}

    To show that the construction works correctly, first assume that one of the two original instances is a yes-instance.
    Since both cases are completely symmetrical, assume that there are shortest disjoint paths between all terminal pairs~$(s_a^i,t_a^i)$ for all~$a \in [k]$ in~$G_i$.
    Then, we can connect all terminal pairs~$(s_b,t_b)$ by using the unique shortest paths between~$s_b$ and~$s_b^i$ and between~$t_b^i$ and~$t_b$ for all~$b \in [k]$ together with the solution paths inside~$G_i$.
    Now assume that there is a solution in the constructed instance, that is, there are pairwise vertex-disjoint shortest paths between all terminal pairs~$(s_b,t_b)$ for all~$b \in [k]$.
    First assume that the~$s_1$-$t_1$-path passes through~$G_i$.
    Then, this path uses the unique shortest path from~$t_1^i$ to~$t_i$.
    Note that this path blocks all paths between~$t_b^j$ and vertices in~$G_j$ for all~$b \neq 1$.
    Thus, all paths have to pass through the graph~$G_i$.
    Note that the only possible way to route vertex-disjoint paths from all~$s$-terminals to all~$s^i$ terminals and from all~$t^i$-terminals to all~$t$-terminals is to connect~$s_a$ to~$s_a^i$ and~$t_a^i$ to~$t_a$ for all~$a \in [k]$.
    This implies that there is a solution that contains vertex-disjoint shortest paths between~$s_a^i$ and~$t_a^i$ in~$G_i$ for all~$a \in [k]$, that is, at least one of the two original instances is a yes-instance.
    The case where the~$s_1$-$t_1$-path passes through~$G_j$ is analogous since the only monotone path from~$s_1$ to a vertex in~$G_j$ is the unique shortest~$s_1$-$s_1^j$-path and this path blocks all monotone paths from~$s_a$ to vertices in~$G_i$ for all~$a \neq 1$.
    
    Note that the constructed graph is layered and that the number of layers is~$\lambda + 2k$.
    Moreover, the size of the new instance is in~$O(|G_i| + |G_j| + k^2)$.
    To complete the reduction, we iteratively half the number of instances by partitioning all instances into arbitrary pairs and merge the two instances in a pair into one instance.
    After~$\log t$ iterations, we are left with a single instance which is a yes-instance if and only if at least one of the~$t$ original instances is a yes-instance.
    The size of the instance is in~$O(\sum_{i \in [t]} |G_i| + t \cdot k^2)$ which is clearly polynomial in~$\sum_{i \in [t]} |G_i|$ as each instance contains at least~$k$ vertices.
    Moreover, the parameter~$\ell$ in the constructed instance is~$k \cdot (\lambda-1) + 2k\log t$, which is polynomial in~$|G_i|+\log t$ for each graph~$G_i$ as~$G_i$ contains at least one vertex in each of the~$\lambda$ layers and at least~$k$ terminal vertices.
    Thus, all requirements of an OR-cross-composition are met and this concludes the proof.
\end{proof}

Note that since \ldsp{} is a special case of \dsp{} (and therefore of \mdsp), \cref{thm:nopkl} also excludes polynomial kernels for these problems parameterized by~$\ell$.

\begin{corollary}
    \dsp{} and \mdsp{} parameterized by $\ell$ do not admit polynomial kernels unless \npconp.
\end{corollary}

\section{Conclusion}
\label{sec:concl}
In this paper, we studied \mdsp.
We show that there is no~$m^{\nicefrac{1}{2}-\varepsilon}$-approximation in polynomial time unless~P $=$ NP.
Moreover, if FPT~$\neq$~W[1] or assuming the stronger gap-ETH, we show that there are no non-trivial approximations for \mdsp{} in $f(k)\poly(n)$ time.
When parameterized by~$\ell$, there is a simple~$\lceil\sqrt{\ell}\rceil$-approximation in polynomial time that matches the $m^{\nicefrac{1}{2}-\varepsilon}$ lower bound as~$\ell \leq \min(n,m)$.
Finally, we showed that \mdsp{} can be solved in~$2^{O(k)}\poly(n)$ time, but it does not admit a polynomial kernel and assuming the ETH, it cannot be solved in~$2^{o(\ell)}\poly(n)$ time.
Also assuming the ETH, it cannot be solved in~$2^{o(\sqrt{n})}$ time on planar graphs.

A way to combine approximation algorithms and the theory of (polynomial) kernels are \emph{lossy kernels}~\cite{LPRS17}.
Since the exact definition is quite technical and not relevant for this work, we only give an intuitive description.
An~$\alpha$-approximate kernel or lossy kernel for an optimization problem is a pair of algorithms that run in polynomial time which are called \emph{pre-processing algorithm} and \emph{solution-lifting algorithm}.
The pre-processing algorithm takes as input an instance~$(I,\rho)$ of a parameterized problem~$P$ and outputs an instance~$(I',\rho')$ of~$P$ such that~$|I'|+\rho' \leq g(\rho)$ for some computable function~$g$.
The solution-lifting algorithm takes any solution~$S$ of~$(I',\rho')$ and transforms it into a solution~$S^*$ of~$(I,\rho)$ such that if~$S$ is a~$\gamma$-approximation for~$(I',\rho')$ for some~$\gamma \geq 1$, then~$S^*$ is an~$\gamma\cdot\alpha$-approximation for~$(I,\rho)$.
If the size of the kernel is~$g(\rho)$ and if~$g$ is constant or a polynomial, then we call it a constant-size or polynomial-size~$\alpha$-approximate kernel, respectively.
It is known that a (decidable) parameterized problem admits a constant-size approximate~$\alpha$-kernel if and only if the unparameterized problem associated with~$P$ can be~$\alpha$-approximated (in polynomial time)~\cite{LPRS17}.
Moreover, any (decidable) parameterized problem admits an~$\alpha$-approximate kernel (of arbitrary size) if and only if the problem can be~$\alpha$-approximated in~$f(\rho)\poly(|I|)$~time.

In terms of lossy kernelization, our results imply that there are no non-trivial lossy kernels for the parameter~$k$.
For the parameter~$\ell$, \cref{prop:lapprox} implies a constant-size lossy kernel for~$\alpha \in \Omega(\sqrt{\ell})$ and \cref{prop:fptl} implies an~$f(\ell)$-size lossy kernels for any~$\alpha \geq 1$. 
This leaves the following gap which we pose as an open problems.

\begin{open problem}
    Are there any $\poly(\ell)$-size lossy kernels for \mdsp{} with~$\alpha \in o(\sqrt{\ell})$ (or even constant $\alpha$)?
\end{open problem}

\ifjourn
\section*{Declarations}

\subparagraph*{Funding.}
This work has received funding from the Research Council of Norway via the project BWCA (grant no. 314528) and the European Research Council (ERC)
under the European Union’s Horizon 2020 research and innovation programme (grant agreement No.
819416).

\subparagraph*{Competing Interests.}
The authors declare no competing interests.

\else
    \bibliographystyle{plainnat}
\fi
\bibliography{ref}

\end{document}